\documentclass[10pt,a4paper]{article}

\usepackage[UKenglish]{babel}
\usepackage{float}

\usepackage{csquotes, enumerate, amsfonts, amssymb, physics, amsmath, xfrac, graphicx, subcaption, placeins, color, booktabs, mathtools, tikz, blkarray, nicematrix, upgreek, tikzscale, bbm, comment, amsthm, microtype, pgfplots, tabularx}

\pgfplotsset{width = \textwidth, compat = 1.9}

\newcommand{\E}[1]{\mathbb{E}\left[#1\right]}
\renewcommand{\P}[1]{\mathbb{P}\left[#1\right]}
\def\dd{\mathrm{d}}
\DeclareMathOperator{\cov}{cov}
\DeclareMathOperator{\vari}{var}
\newcommand{\stir}[2]{\genfrac{[}{]}{0pt}{}{#1}{#2}}

\newtheorem{prop}{Proposition}[section]
\renewcommand{\theprop}{\arabic{prop}}
\newtheorem{coro}{Corollary}[section]
\renewcommand{\thecoro}{\arabic{coro}}
\newtheorem{lem}{Lemma}[section]
\renewcommand{\thelem}{\arabic{lem}}
\theoremstyle{definition}
\newtheorem*{rem}{Remark}

\usepackage[labelfont=bf,labelsep=period]{caption}
\usepackage[sc]{mathpazo}

\usepackage[backend=biber, citestyle=nature]{biblatex}
\usepackage[draft]{todonotes}

\usepackage{authblk}
\usepackage[
 breaklinks=true
 ,colorlinks=true
 ,linkcolor=blue
 ,citecolor=blue
 ,backref=none
 ,pagebackref=false
 ]{hyperref} 

\title{Single-cell mutational burden distributions in birth-death processes}

\author{Christo Morison$^{1}$, Dudley Stark$^{1}$ \& Weini Huang$^{1,2,*}$ }

\affil{
{\footnotesize
$^{1}$School of Mathematical Sciences, Queen Mary University of London, London, E1 4NS United Kingdom.
$^{2}$Group of Theoretical Biology, Research Section of Genomics , School of Life Sciences, Sun Yat-sen University, Guangzhou, 510060 China.\\
$^*$Corresponding author: \texttt{weini.huang@qmul.ac.uk}}
}

\begin{document}
\maketitle

\begin{abstract}

\noindent Genetic mutations are footprints of tumour growth. While mutation data in bulk samples has been used to infer evolutionary parameters hard to measure \textit{in vivo}, the advent of single-cell data has led to strong interest in the mutational burden distribution (MBD) among tumour cells. We introduce dynamical matrices and recurrence relations to integrate this single-cell MBD with known statistics, and derive new analytical expressions. Surprisingly, we find that the shape of the MBD is driven by cell lineage-level stochasticity rather than by the distribution of mutations in each cell division.
\end{abstract}


\textbf{Keywords:} Birth-death process,  mutational burden distribution, single cell statistics, mutation accumulation, site frequency spectrum, infinite sites approximation, cancer evolution.\vspace{-0.4cm}


\section*{Introduction}
Somatic mutations are important for the evolution of biological systems with clonal reproduction, including the development from healthy tissues to cancer~\cite{weinberg2013biology, reuschEvolutionSomaticGenetic2021a}. While less is known about the somatic mutation rates in clonal species such as plants and corals, they have been studied extensively in human tissues. Healthy cells may accumulate in the order of $1$ to $2$ mutations per cell per division, which is directly observable in early development~\cite{baeDifferentMutationalRates2018a, lee-sixPopulationDynamicsNormal2018a, wernerMeasuringSingleCell2020a}. The mutational rate of tumour cells is often thought to be higher, which can be caused for example by genomic instability~\cite{frankProblemsSomaticMutation2004a, komarovaCancerAgingOptimal2005a, burrellCausesConsequencesGenetic2013a}. This large number of mutations accumulated in tumours serves as a genetic footprint to reveal their evolutionary history. Since the majority of these mutations are neutral~\cite{bozicAccumulationDriverPassenger2010a}, not impacting the fitness of a cell compared to its parental cell, neutral theory has been used to explain mutational patterns in many patient samples across different tumour types~\cite{lingExtremelyHighGenetic2015a, sottorivaBigBangModel2015a, williamsIdentificationNeutralTumor2016a}. These measurements often demonstrate an early expansion of tumour populations, wherein driver mutations are clonal and the intratumour heterogeneity arises from neutral passenger mutations accumulated after cancer initiation. Although clonal interference, where cells carrying different sets of driver mutations intercompete, is a likely alternative scenario especially in large populations~\cite{parkClonalInterferenceLarge2007a, karlssonDeterministicEvolutionStringent2023a}, here we focus on a further understanding of mutation accumulation under neutral selection as an important baseline dynamics.

Distributions of genetic heterogeneity under neutral selection have been studied in population genetics for over half a century~\cite{ewensPseudotransientDistributionIts1964a, kimuraGeneticVariabilityMaintained1968a}. One of such statistics is the site frequency spectrum (SFS), which describes the frequencies of mutations in a population~\cite{fuStatisticalPropertiesSegregating1995a}. Because the SFS deals with population-level information, it can be compared to bulk genomic data or pooled singe-cell data ~\cite{williamsIdentificationNeutralTumor2016a, bozicQuantifyingClonalSubclonal2016a,moellerMeasuresGeneticDiversification2022}. For an exponentially-growing population, a rescaling of the SFS, the variant allele frequency spectrum, has been shown to follow a $1 / f^2$ power law, for $f$ the frequency of a mutation in the population~\cite{durrettPopulationGeneticsNeutral2013, durrettBranchingProcessModels2015a, bozicQuantifyingClonalSubclonal2016a, williamsIdentificationNeutralTumor2016a}. More recently, exact expressions for the SFS were found under the assumption of neutral evolution~\cite{gunnarssonExactSiteFrequency2021a}.

The advent of single-cell sequencing~\cite{shapiroSinglecellSequencingbasedTechnologies2013a, wangAdvancesApplicationsSingleCell2015a} opens the door for combining bulk and single-cell data to understand the growth history and dynamic traits of (healthy or tumorous) tissues, which are otherwise difficult to measure directly  \cite{abascalSomaticMutationLandscapes2021a, moellerMeasuresGeneticDiversification2022}. There is great need for new mathematical and computational machinery to cope with single-cell data, which provides different mutational distributions beyond the SFS. The number of unique mutations in the population, also known as the overall tumour mutational burden (TMB)~\cite{chalmersAnalysis1000002017a}, has been studied both in a single tumour~\cite{gunnarssonExactSiteFrequency2021a} and (its distribution) between tumours~\cite{fernandezCancerSpecificThresholdsAdjust2019a, martinez-perezPanelsModelsAccurate2021a}. However, the distribution of mutational burdens between cells, the so-called single-cell mutational burden distribution (MBD), has only recently been experimentally observable through single-cell sequencing. Understanding the MBD may further help in inferring important evolutionary parameters, determining the growth history of the tumour and the level of selection at play, with neutral selection as a baseline with which to compare. Using data from healthy haematopoietic stem cells and oesophageal epithelial cells, Moeller \& Mon P\`ere \textit{et al.}~showed that analysis of single-cell and bulk data complement each other and narrowed down the parameter inference of the mutation rate and stem cell population size~\cite{moellerMeasuresGeneticDiversification2022}. More specifically, the mean and variance of the MBD for a growing population were derived and used to estimate the underlying mutation rate~\cite{moellerMeasuresGeneticDiversification2022}. However, the exact analytical shape of the MBD has not yet been explicitly found.

The MBD evolves during the cell division process, and thus an instructive object to study it with is the cell lineage tree~\cite{derenyiHowMutationAccumulation2024}, whose leaves symbolise living cells and whose root is the progenitor of the population. Branching processes can then be viewed as growing trees, where cell division is represented by a leaf bifurcating into two leaves, and cell death is the removal of a leaf. Because this framework can generate phylogenetic trees~\cite{pageMolecularEvolutionPhylogenetic2009}, cell lineage trees have properties that have been extensively studied~\cite{steelPropertiesPhylogeneticTrees2001a}. One such property is the distribution of leaf heights (or the distances in edges from root to leaves), known as the division distribution (DD) of individual cells. By including the accumulation of new mutations at internal nodes of the tree, the MBD is obtained~\cite{derenyiHowMutationAccumulation2024}. An expression for the DD generated by a pure-birth, or Yule, process has been found~\cite{lynchMoreCombinatorialProperties1965a}; though when death is included it has not yet been solved exactly. Our goal is to build upon knowledge of the DD and the SFS to better understand the MBD, by formulating an discrete-time approach that integrates all three distributions.

We introduce a new framework via dynamical matrices to investigate mutation accumulation in a birth-death process and explain how key quantities such as the SFS, DD and MBD are obtained from these mutational matrices. This framework allows us to derive exact solutions of these distributions by recurrence relations in the pure-birth case, as well as first-order approximations when death is introduced, which hold in the low-death and large-population limits. By comparing our solutions for the SFS and DD to known results in population genetics, we first demonstrate the efficacy of our framework.
We then show new results in expressions for the birth-death DD (equation~\eqref{birthdeathddapprox}) and both the pure-birth and birth-death MBD (equation~\eqref{poissonificationexpectationseq3}). Our analytical results for all three distributions agree well with stochastic simulations. We find that the MBD can be generated via the DD and the mean mutation rate per cell division, independently of the stochasticity in the number of mutations per cell division.

\section*{Model}

\subsection*{Mutation accumulation in a birth-death process}

In a birth-death process where either a cell divides with probability $\beta$ or dies with probability $\delta = 1 - \beta$, the population size $N_i$ at time step $i$ can be described by a discrete-time Markov chain (Figure~\ref{chain-tree-fig}a). The state space in this Markov chain is the finite integer set $\{0, \dots, N\}$ , where $N$ is the largest possible population size. In some cases, we are interested in the limit $N \to \infty$.

Often, the stochastic birth-death process explicitly involves a continuous time parameter $t$ instead of a discrete step count $i$. This allows for rates to be considered instead of probabilities; however, as long as events are assumed not to be simultaneous, these two schemes can be mapped to one another by choosing a distribution of times between events. Most often, events are assumed to be exponentially distributed, and thus their frequency grows with the population size.

\begin{figure}[ht]
    \centering
    \includegraphics[width = 0.9\textwidth]{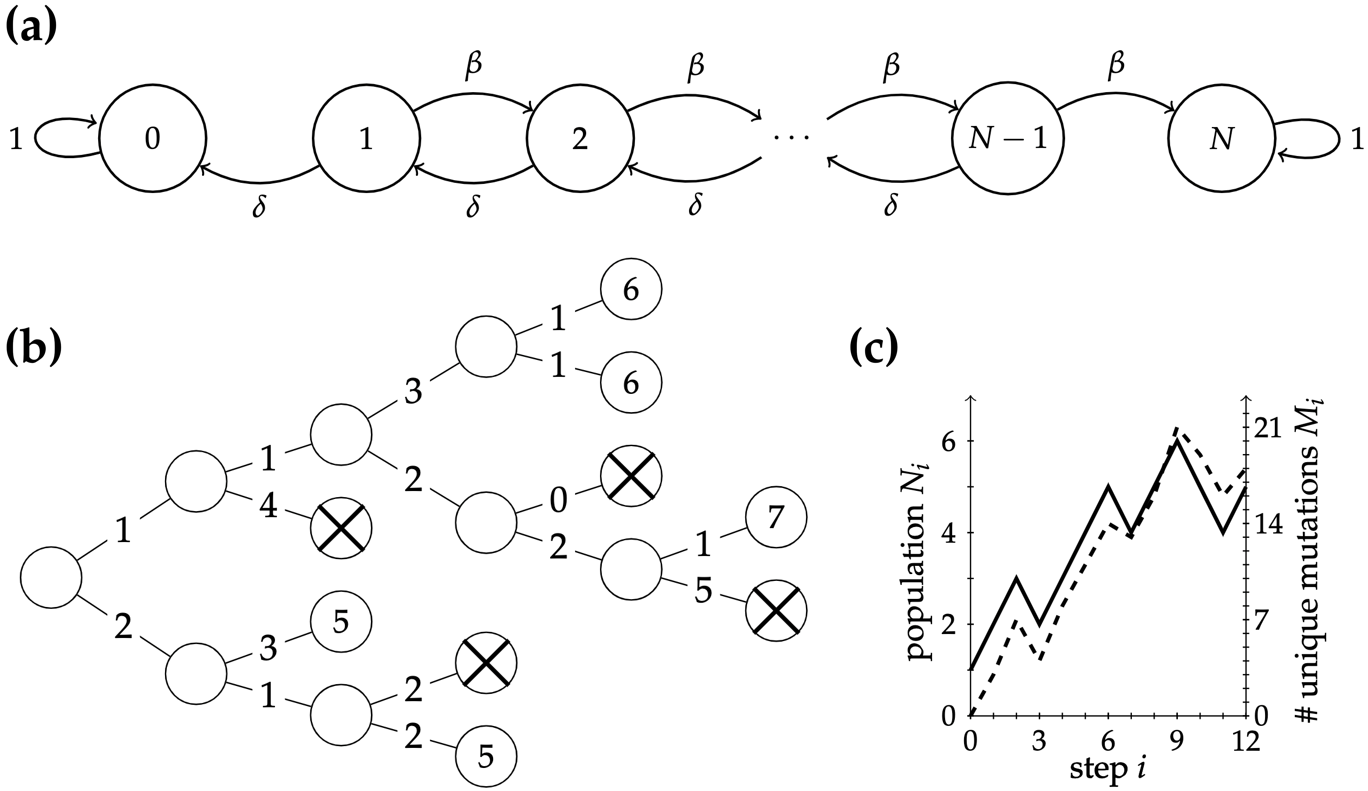}
    \caption{\textbf{(a)}~Discrete-time Markov chain description of the population size. \textbf{(b)}~Growing binary tree representation of an example realisation of the birth-death process with mutations described in the main text, with birth probability $\beta = 2 / 3$, death probability $\delta = 1 / 3$, mutational mean $\mu = 2$, initial population $N_0 = 1$ and final population $N_{12} = 5$. Cells that are crossed out have died. Edges are labelled by the number of new mutations occurring during that division. Leaves (living cells) are labelled by their mutational burden, which is equal to the sum of the edges that connect them to the root, or the mutation-free progenitor cell. \textbf{(c)}~The population size $N_i$ (solid line) and number of unique mutations $M_i$ (dashed line) versus the step count $i$ for the example realisation in~(b).}
    \label{chain-tree-fig}
\end{figure}

Here, we focus on the growing-population case $\beta > \delta$ and assume $N_0 = 1$ unless otherwise mentioned. In this case, the birth-death process is a growing rooted binary tree, where the root is the lone progenitor cell (assumed to be mutation-free, as any of its mutations will be clonal in the population), leaves are living cells, and pruned leaves are dead cells. See Figure~\ref{chain-tree-fig}b for an example realisation of such a process. Novel mutations are accumulated during cell divisions and old mutations may be lost in cell death. When a cell divides, its daughter cells inherit all mutations carried by the mother cell and acquire a random number of new mutations drawn from a Poisson distribution with mean $\mu$: $U_1, U_2 \sim \text{Pois}(\mu)$, where the indices refer to the two daughter cells. These mutations are unique under the infinite sites approximation, where the probability of two point mutations occurring at the same location along the large genome is vanishingly small~\cite{kimuraNumberHeterozygousNucleotide1969a}. The point mutations occurring during each duplication of the genome are thought of independent from each other and thus their number follows a Poisson distribution~\cite{lederbergReplicaPlatingIndirect1989a, zhengMathematicalIssuesArising2003a}.

We are most interested in three key quantities of this birth-death process with mutations at each step $i$: \textit{(i)}~the site frequency spectrum (SFS) $\{S_{j, i}\}_j$, whose elements $S_{j, i}$ denote the number of mutations which occur $j$ times in the population~\cite{fuStatisticalPropertiesSegregating1995a}; \textit{(ii)}~the single-cell mutational burden distribution (MBD) $\{B_{k, i}\}_k$, whose elements $B_{k, i}$ are the number of cells having a mutational burden of $k$~\cite{moellerMeasuresGeneticDiversification2022}; and \textit{(iii)}~the division distribution (DD) $\{D_{\ell, i}\}_\ell$, whose elements $D_{\ell, i}$ indicate the number of cells having undergone $\ell$ divisions during the process. Note that $D_{\ell, i}$ is also the number of leaves lying at a distance of $\ell$ edges from the root in the growing tree framework.

While the importance of the SFS and the DD have been investigated in growing populations~\cite{steelPropertiesPhylogeneticTrees2001a, durrettPopulationGeneticsNeutral2013, durrettBranchingProcessModels2015a, williamsIdentificationNeutralTumor2016a, bozicQuantifyingClonalSubclonal2016a, gunnarssonExactSiteFrequency2021a}, we are interested in the relationship between them and how it can help us understand the MBD. Here, we introduce a novel discrete-time framework to demonstrate the symmetry between these distributions. The number of unique mutations at step $i$ is $M_i = \sum_jS_{j, i}$, and the population size is $N_i = \sum_kB_{k, i} = \sum_\ell D_{\ell, i}$, both of which are plotted in Figure~\ref{chain-tree-fig}c for the example found in Figure~\ref{chain-tree-fig}b. Thus, the MBD and the DD form partitions of the number of cells in a way similar to the SFS partitioning the number of unique mutations. Next, we introduce dynamical matrices to connect those quantities arising from the same population of individual cells. 

\subsection*{Dynamical matrices to unite the SFS, DD and MBD}

We consider a collection of matrices $Y_i$, where the rows refer to cells and the columns refer to mutations, known as genotype matrices or SNP (single nucleotide polymorphism) matrices in bioinformatics~\cite{ronenLearningNaturalSelection2013a}. Our matrices are dynamical in that their entries are updated at each step by a binary filling in the following manner: the $(n, m)$th entry of the matrix is equal to $1$ if the $n$th cell possesses the $m$th mutation at step $i$ and equal to $0$ otherwise. When a cell dies, its row is removed from the matrix. Figure~\ref{mut-matrix-fig}a shows an example of the matrix $Y_i$ associated to the tree example in Figure~\ref{chain-tree-fig}b. We extend the concept of genotype matrices by marking mutations arising during a single (past) division by grey shaded areas. Note that if no mutations arise during a division, the corresponding $0$ entry would still be shaded in grey, as this shading tracks the division burden. These matrices are also how mutational data can be stored in stochastic simulations. 

\begin{figure}[ht]
    \centering
    \includegraphics[width = \textwidth]{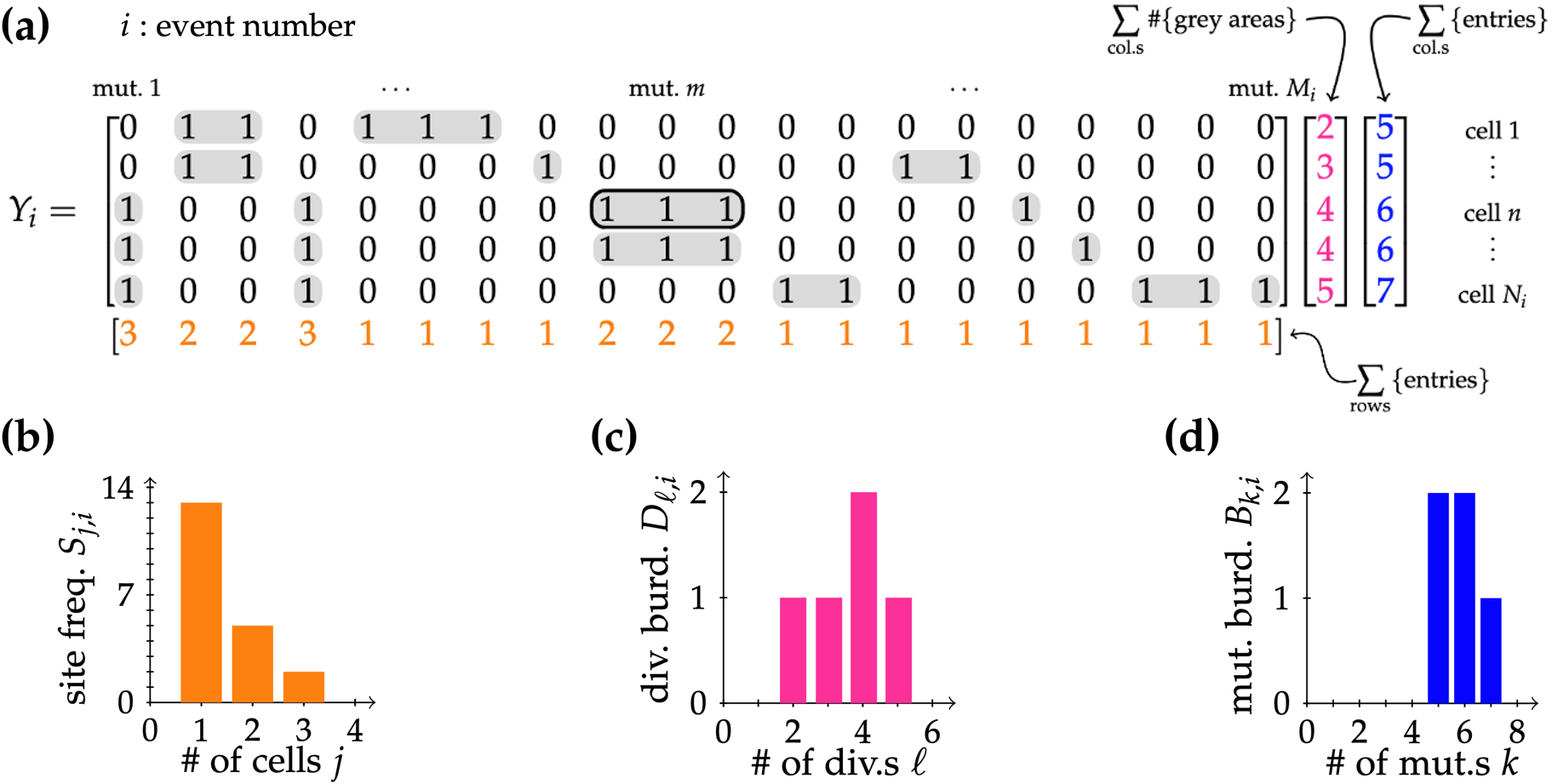}
    \caption{The matrix framework described in the main text, where $i$ refers to the Markov step count (in this example, $i = 12$). \textbf{(a)}~The matrix $Y_i$ corresponding to the example realisation of the birth-death process depicted in Figure~\ref{chain-tree-fig}b, where entry $(n, m)$ is $1$ if cell $n$ possesses mutation $m$ and $0$ otherwise. Grey shaded mutational entries arose during the same division and thus always occur together in the descendants of their progenitor. For example, mutations $m - 1$, $m$ and $m + 1$ in cell $n$ (shown outlined in black) were generated in the same past division. (In Figure~\ref{chain-tree-fig}b, we can determine by inspection this ancestor cell to be the one that later divided into two cells with mutational burdens of $6$.) The row sum of the entries of $Y_i$ is shown in orange, the column sum of the number of grey areas is in pink and the column sum of the entries of $Y_i$ is in blue. \textbf{(b)}~The site frequency spectrum (SFS) $\{S_{j, i}\}_j$: a histogram of the bottommost (orange) vector of~(a). \textbf{(c)}~The division distribution (DD) $\{D_{\ell, i}\}_\ell$: a histogram of the middle (pink) vector of~(a). \textbf{(d)}~The single-cell mutational burden distribution (MBD) $\{B_{k, i}\}_k$: a histogram of the rightmost (blue) vector of~(a), or equivalently of the sums of the weights of the edges in Figure~\ref{chain-tree-fig}b from the root to the leaves.}
    \label{mut-matrix-fig}
\end{figure}

We can obtain the distributions of our key quantities, the SFS, DD and MBD (Figures~\ref{mut-matrix-fig}b--d), from our dynamical mutation matrices (Figure~\ref{mut-matrix-fig}a). For each mutation (column), the number of cells carrying this mutation is the row sum of the entries of $Y_i$ (orange vector in Figure~\ref{mut-matrix-fig}a). Thus, the histogram of this vector is the SFS at step $i$. For each cell (row), the number of divisions that the cell has undergone is the column sum of the number of grey areas (pink vector in Figure~\ref{mut-matrix-fig}a), and the number of mutations in this cell is the column sum of the entries of $Y_i$ (blue vector in Figure~\ref{mut-matrix-fig}a). Correspondingly, the histograms of these two vectors lead to the other distributions obtainable from single-cell information: the DD and the MBD.

The symmetry provided by this mutation matrix $Y_i$ gives rise to the following relationship between the site frequency spectrum and the single-cell mutational burden distribution:
\begin{equation}\label{sfsmbdmirror}
    \sum_{j = 1}^{N_i}jS_{j, i} = \sum_{k = 0}^{M_i}kB_{k, i}.
\end{equation}
We call this quantity the number of mutational occurrences: that is, the sum of the entries of $Y_i$.

\begin{table}[ht]
    \centering
    \begin{tabularx}{0.85\textwidth}{cX}
    \toprule
        \textbf{Symbol} & \textbf{Description}\\\hline
        $\beta$ & Birth probability (Markov transition rate)\\
        $\delta$ & Death probability (Markov transition rate)\\
        $N_i$ & Population size at step $i$\\
        $N$ & Maximal population size\\
        $\mu$ & Mutational mean: mean number of mutations acquired per division per daughter cell\\
        $M_i$ & Number of unique mutations (tumour mutational burden, TMB) at step $i$\\
        $\{S_{j,i}\}_j$ & Site frequency spectrum (SFS), where $S_{j,i}$ is the number of mutations occurring $j$ times in the population at step $i$\\
        $\{B_{k, i}\}_k$ & Single-cell mutational burden distribution (MBD), where $B_{k, i}$ is the number of cells with $k$ mutations at step $i$\\
        $\{D_{\ell, i}\}_\ell$ & Division distribution (DD), where $D_{\ell, i}$ is the number of cells having undergone $\ell$ divisions at step $i$\\
        $Y_i$ & Mutational matrix at step $i$: entry $(n, m)$ is $1$ if cell $n$ possesses mutation $m$, $0$ otherwise\\
        $\stir{i}{\ell}$ & Unsigned Stirling number of the first kind with indices $1 \leq \ell \leq i$\\
        \bottomrule
    \end{tabularx}
    \caption{Notation used in this manuscript.}
    \label{notationtab}
\end{table}

\section*{Results}

Our primary approach for deriving the distributions of our key quantities from the discrete-time model is as follows. We use the law of total expectation ($\E{X} = \E{\E{X\,\middle|\,Y}}$, for any random variables $X$ and $Y$) to equate an expected quantity at step $i + 1$ to a conditional expectation. This usually is a function of the quantity at step $i$ conditional on knowledge at step $i$, as earlier knowledge is never needed due to the Markov nature of the model. From this, we derive a recurrence relation for the expected values of our desired quantity, which can be solved.

We first note that conditioning on the survival of the entire population plays a role in all of our subsequent expected values. In the pure-birth case, the population is deterministic and equal to $N_i = i + 1$. Once death is included, however, the population size becomes a random variable. For the birth-death chain of Figure~\ref{chain-tree-fig}a, its expected value at step $i$, both conditioned on survival and not, can be exactly calculated, which is done in Proposition~\ref{expectedniprop} of the \hyperlink{si}{Supplementary Information}. Figure~\ref{n-fig} shows that the expected population size conditioned on survival can be linearly approximated by $\E{N_i} \simeq (\beta - \delta)i + 1$, valid for low death, since in this limit the expected gain in population in one step is $\beta - \delta$. All of our ensuing expectations are conditioned on survival and the initial condition $N_0 = 1$, which we will omit from our notation for brevity.

\subsection*{Site frequency spectrum}

With the recurrence relation method outlined above, we can formally derive the pure-birth ($\beta = 1$) site frequency spectrum. When there is no death, $N_i = i + 1$, and thus the expected number of $j$-abundant mutations in the dividing cell at step $i$ is $j S_{j, i} / (i + 1)$. After the division in step $i$, any $j$-abundant mutations in the dividing cell will become $(j + 1)$-abundant and thus no longer contribute to the $j$-site. Similarly, $(j - 1)$-abundant mutations in the dividing cell now contribute to the $j$-site.  We therefore have
\begin{equation}\label{purebirthsfsrr}
    \E{S_{j, i + 1}} = \E{\E{S_{j, i + 1}\,\middle|\,\{S_{j', i}\}_{j'}}} = \E{S_{j, i} - \frac{j S_{j, i}}{i + 1} + \frac{(j - 1) S_{j - 1, i}}{i + 1} + \left(U_1 + U_2\right)\delta_{1, j}},
\end{equation}
for $\delta_{\cdot, \cdot}$ the Kronecker delta symbol, whose source term arises from the new $1$-abundant mutations occurring during division, which are drawn from a Poisson distribution: $U_1, U_2 \sim \text{Pois}(\mu)$. We make the change of variables $Q_j = \E{S_{j, i}} / (i + 1)$, whose lack of dependence on $i$ can be computed by brute force from equation~\eqref{purebirthsfsrr}, absorbing the source term into the boundary condition $Q_1 = \mu$. Using the linearity of expectation, equation~\eqref{purebirthsfsrr} becomes simply $(j + 1)Q_j = (j - 1)Q_{j - 1}$, which telescopes to obtain the known result
\begin{equation}\label{purebirthsfs}
    \E{S_{j, i}} = \frac{2\mu(i + 1)}{j(j + 1)}.
\end{equation}
An identical procedure can be applied in the birth-death case to recover the large-population expected SFS, though a difficulty here is that now $N_i$ becomes a random variable itself. Since to first order $\E{A / B} \simeq \E{A} / \E{B}$ (when $A$ and nonzero $B$ are close to their expected values), our first-order approximation (whose derivation is found in Proposition~\ref{sfsprop} of the \hyperlink{si}{Supplementary Information}), which matches the $N \to \infty$ result from~\cite{gunnarssonExactSiteFrequency2021a}, is 
\begin{equation}\label{birthdeathsfsapprox}
    \E{S_{j, i}} \simeq \sum_{j' = 0}^\infty \frac{2\mu(\delta / \beta)^{j'}}{(j + j')(j + j' + 1)}\E{N_i},
\end{equation}
where all expectations are conditioned on non-extinction of the whole population~\cite{gunnarssonExactSiteFrequency2021a}. In the limit of low death ($\delta \ll \beta$), this approximation is sound, as then the variance in population size is small (see Proposition~\ref{varianceprop} for details).

\subsection*{Division distribution}

The expected division distribution in the pure-birth case can be obtained in a similar manner as the site frequency spectrum. The probability of selecting a cell with $\ell$ divisions in its history is $D_{\ell, i} / (i + 1)$, which will then no longer contribute to $D_{\ell, i + 1}$, dividing into two cells with one more division in their history than before. The law of total expectation then becomes
\begin{equation}\label{purebirthddrr}
    \E{D_{\ell, i + 1}} = \E{\E{D_{\ell, i + 1}\,\middle|\,\{D_{\ell', i}\}_{\ell' \leq i}}} = \E{D_{\ell, i} - \frac{D_{\ell, i}}{i + 1} + \frac{2 D_{\ell - 1, i}}{i + 1}},
\end{equation}
which can be solved to recover the result from~\cite{lynchMoreCombinatorialProperties1965a}:
\begin{equation}\label{purebirthdd}
    \E{D_{\ell, i}} = \stir{i}{\ell}\frac{2^\ell}{i!},
\end{equation}
where the unsigned Stirling numbers of the first kind $\stir{i}{\ell}$ are defined by the relation
\begin{equation}\label{stirlingdef}
    \stir{i + 1}{\ell} = i\stir{i}{\ell} + \stir{i}{\ell - 1} \quad \text{ for } 1 \leq \ell \leq i,
\end{equation}
with boundary conditions $\stir{1}{1} = 1$ and $\stir{i}{\ell} = 0$ if $\ell > i$ or $\ell = 0$. Results from stochastic simulations using a Gillespie algorithm agree well with this prediction (see Figure~\ref{dd-fig}).

\begin{figure}[ht]
    \centering
    \includegraphics[width = 0.65\textwidth]{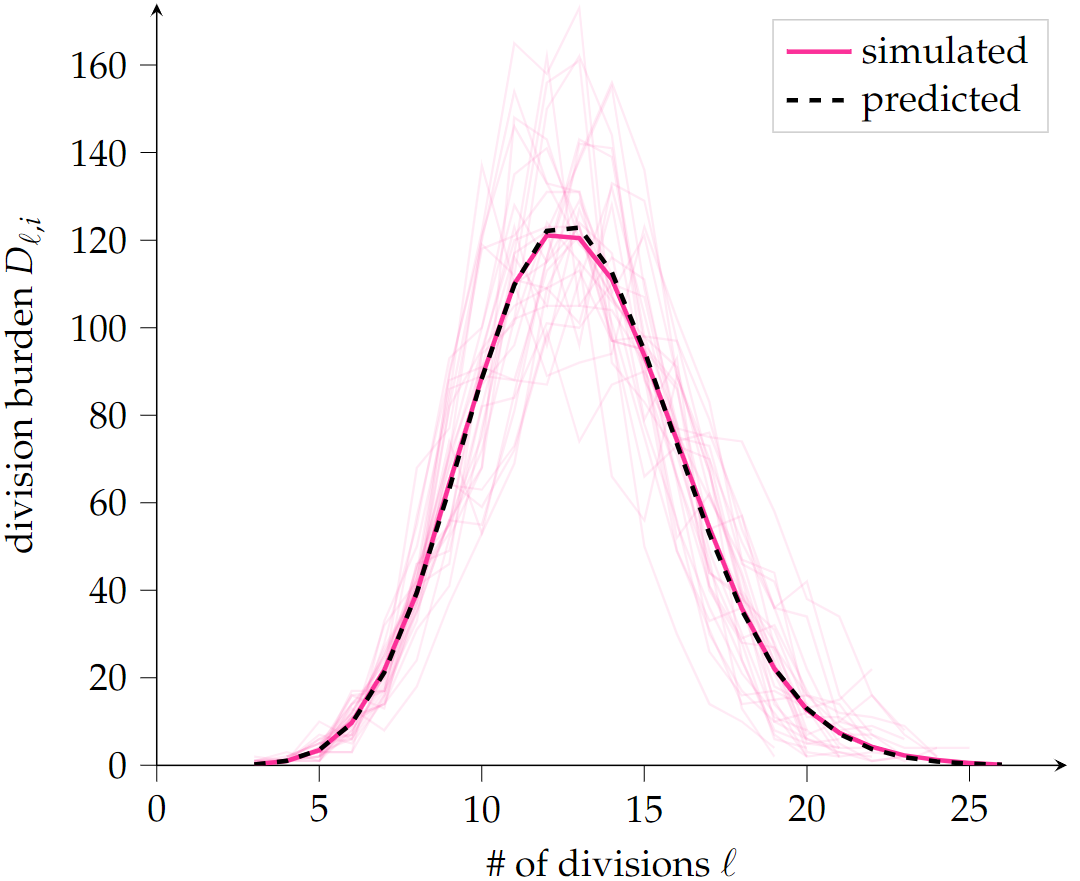}
    \caption{Average (solid dark pink line) of $200$ simulation realisations (representatives in solid pale pink lines) of the division distribution (DD) for a pure-birth process up to final population size $N = 10^3$, along with the predicted expected DD obtained from equation~\eqref{purebirthdd} (dashed black line).}
    \label{dd-fig}
\end{figure}

Similarly to with the SFS, in the $\delta > 0$ case we can obtain a first-order approximation (see Proposition~\ref{ddfirstorderprop}) for the DD:
\begin{equation}\label{birthdeathddapprox}
    \E{D_{\ell, i}} \simeq \frac{\stir{i}{\ell}2^\ell\left(1 - \delta / \beta\right)^{-\ell}}{\sum_{\ell' = 1}^i\stir{i}{\ell'}2^{\ell'}\left(1 - \delta / \beta\right)^{-\ell'}}\E{N_i},
\end{equation}
wherein it is evident that the division distribution partitions the population size (since summing the fraction on the right-hand side over $\ell$ gives unity) and that this partitioning is orchestrated by the functions $f(i, \ell) = \stir{i}{\ell}2^\ell\left(1 - \delta / \beta\right)^{-\ell}$.

\subsection*{Mutational burden distribution}

The single-cell mutational burden distribution differs from the division distribution because there is an additional stochasticity at each cell division due to mutational (Poisson) distributions. To obtain an expected MBD from a DD, we can employ a procedure to introduce this stochasticity as follows: each cell contributing to a division burden $D_{\ell, i}$ will have undergone $\ell$ divisions, so will have acquired $\sum_{p = 1}^\ell U_p$ mutations, where $U_p \sim \text{Pois}(\mu)$ represents the number of mutations acquired during the cell's $p$th division. Since the Poisson distribution is additive, this sum is in turn a Poisson-distributed random variable with mean $\ell\mu$. The left-hand side of Figure~\ref{mbd-fig}a qualitatively depicts the elements of the DD being converted into Poisson probability mass functions associated to these sums of Poisson variables. These are then summed to obtain the MBD, shown on the right-hand side of Figure~\ref{mbd-fig}a, in the following manner.

Writing $U_{p, q, \ell} \sim \text{Pois}(\mu)$ for the number of mutations acquired during the $p$th division of the $q$th cell (for some labelling of cells $1 \leq q \leq D_{\ell, i}$) having undergone $\ell$ divisions, we can sum over the elements of the DD labelled by $\ell$ to obtain
\begin{equation}\label{poissonificationexpectationseq1}
    \E{B_{k, i}} = \E{\E{B_{k, i}\,\middle|\,\{D_{\ell', i}\}_{\ell' \leq i}}} = \E{\E{\sum_{\ell = 0}^i\sum_{q = 1}^{D_{\ell, i}}\mathbbm{1}_{\left\{\sum_{p = 1}^\ell U_{p, q, \ell} = k\right\}}\,\middle|\,\{D_{\ell', i}\}_{\ell' \leq i}}},
\end{equation}
where we have used the indicator function $\mathbbm{1}_A$ to be $1$ on the set $A$ and $0$ elsewhere.

\begin{figure}[H]
    \centering
    \includegraphics[width = \textwidth]{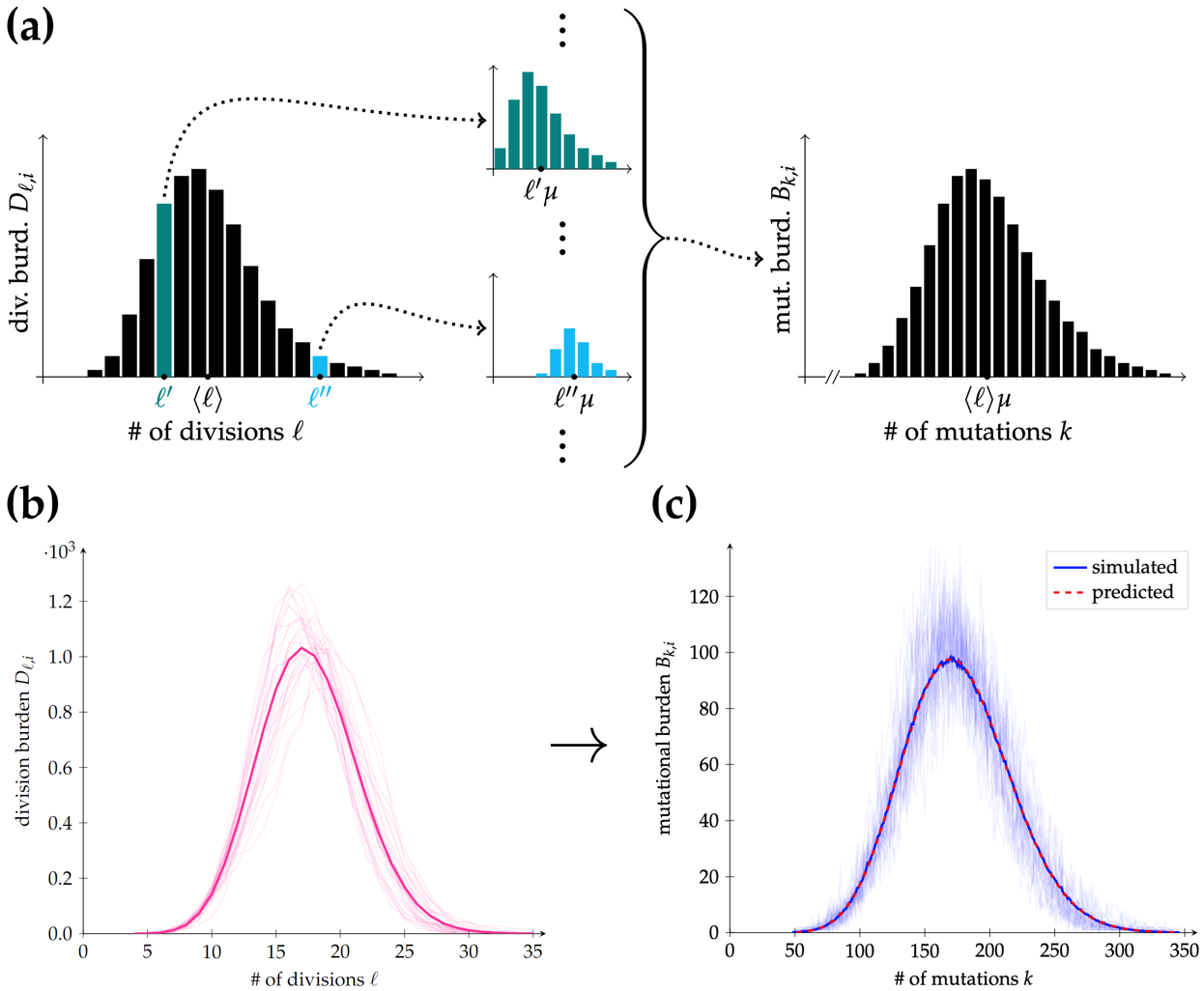}
    \caption{An illustration of the conversion of a division distribution (DD) into a single-cell mutational burden distribution (MBD). \textbf{(a)}~The elements of a DD are translated into Poisson distributions with means $\ell\mu$, weighted by $D_{\ell, i}$ (such that the sums over the teal and cyan distributions are $D_{\ell', i}$  and $D_{\ell'', i}$, respectively, for example), and then summed to obtain the corresponding MBD. Note that if the mean of the DD is $\langle\ell\rangle$, then the mean of the resulting MBD is $\langle\ell\rangle\mu$. \textbf{(b)}~Average (solid dark pink line) of $200$ simulation realisations (representatives in solid pale pink lines) of the DD for a pure-birth process up to final population size $N = 10^4$ with mutational mean $\mu = 10$. \textbf{(c)}~Average (solid dark blue line) of the MBD for the same simulation realisations as~(b) (representatives in solid pale blue lines), along with the MBD obtained from converting the average DD as explained in~(a) and the main text (dashed red line).}
    \label{mbd-fig}
\end{figure}

Now, using the linearity of expectation, that $i$ is fixed and $\ell$ just an index, and the independence of the random variables $D_{\ell, i}$ and $U_{p, q, \ell}$, the right-hand side of equation~\eqref{poissonificationexpectationseq1} becomes
\begin{equation*}\label{poissonificationexpectationseq2}
    \sum_{\ell = 0}^i\E{D_{\ell, i}\E{\mathbbm{1}_{\left\{\sum_{p = 1}^\ell U_{p, q, \ell} = k\right\}}}} = \sum_{\ell = 0}^i\E{D_{\ell, i}}\E{\mathbbm{1}_{\left\{\sum_{p = 1}^\ell U_{p, q, \ell} = k\right\}}}.
\end{equation*}
Finally, substituting the expression in equation~\eqref{purebirthdd} for the expected DD and the probability mass function for the Poisson distribution with mean $\ell\mu$, we find the pure-birth expected MBD:
\begin{equation}\label{poissonificationexpectationseq3}
    \E{B_{k, i}} = \sum_{\ell = 0}^i\stir{i}{\ell}\frac{2^\ell}{i!}\frac{\mathrm{e}^{-\ell\mu}(\ell\mu)^k}{k!}.
\end{equation}
Figures~\ref{mbd-fig}b--c verify the conversion from a DD to a MBD described in the previous discussion with simulations.

The same conversion procedure can be implemented in the birth-death case. Again, working with first-order approximations, the expression in equation~\eqref{birthdeathddapprox} for the expected birth-death DD can be used instead of the pure-birth expression in equation~\eqref{purebirthdd} during the final step to obtain a first-order approximation of the expected birth-death MBD.

Finally, consider the number of mutational occurrences: that is, the sum of the entries of the mutational matrix $Y_i$ or, equivalently, either side of equation~\eqref{sfsmbdmirror}. If this quantity is divided by the number of mutations $M_i$, we obtain the mean of the SFS; and if it is divided by the population $N_i$, we obtain the mean of the MBD. We can derive the expected number of mutational occurrences using our recurrence relation approach, from which we deduce that this mean, representing the expected mutational burden of a cell, 
grows logarithmically with the step count $i$ (see Propositions~\ref{umbeta1prop} and~\ref{umprop}). In the pure-birth case, it is simply a rescaling of the harmonic numbers.

\section*{Discussion}

The distribution of genetic mutations in cell populations has been studied both in the cases of constant~\cite{fuStatisticalPropertiesSegregating1995a, durrettProbabilityModelsDNA2008, ronenLearningNaturalSelection2013a, moellerMeasuresGeneticDiversification2022} and growing populations~\cite{durrettBranchingProcessModels2015a, williamsIdentificationNeutralTumor2016a, simonsDeepSequencingProbe2016a, loebExtensiveSubclonalMutational2019a, watsonEvolutionaryDynamicsFitness2020a, poonSynonymousMutationsReveal2021a, tungSignaturesNeutralEvolution2021a, kurpasModesSelectionTumors2022a}. With the development of single-cell sequencing technologies, exploration of more precise information in single cells is sure to follow in the footsteps of population-level research~\cite{shapiroSinglecellSequencingbasedTechnologies2013a, wangAdvancesApplicationsSingleCell2015a, choComparisonCellState2022}. At the population level, both site frequency spectra (SFS) and overall tumour mutational burden (TMB) have been investigated analytically~\cite{durrettPopulationGeneticsNeutral2013, durrettBranchingProcessModels2015a, bozicQuantifyingClonalSubclonal2016a, williamsIdentificationNeutralTumor2016a, gunnarssonExactSiteFrequency2021a, moellerMeasuresGeneticDiversification2022}. Here we focus on the single-cell distribution of the latter (the single-cell mutational burden distribution, or MBD), and use the foundation of the SFS to better understand the MBD analytically. 

A new framework uniting the SFS and the MBD is presented, relying on a simple procedure: dynamical matrices store the mutational information of a population of cells, whose size is dictated by a birth-death process. Our approach of encoding the data in binary matrices, where the entry $(n, m)$ is $1$ when cell $n$ has mutation $m$ and $0$ otherwise, naturally emerges from the (neutral) evolution-motivated idea wherein a cell is identified by its mutation load~\cite{ronenLearningNaturalSelection2013a}. Two different ways of partitioning the entries of this mutational matrix provide definitions of both the SFS and the MBD as histograms of the row- and column-sums, respectively, as shown in Figure~\ref{mut-matrix-fig}. With this symmetry in mind, which gives rise to equation~\eqref{sfsmbdmirror}, an identical analytical approach depending on the discrete-time Markov nature of the model can be applied to both cases, along with an intermediary case of the division distribution (DD), to obtain recurrence relations for the distributions of interest: we employ the law of total expectation to write the expected value of a quantity of interest in terms of expected values at the previous time step. These recurrences are solved exactly in the pure-birth case and approximately in the birth-death case, giving rise to analytical predictions for the SFS, DD and MBD, which are compared to stochastic simulations as well as previous work on the SFS and the DD.

Indeed, our model recovers the expected values of the SFS and TMB (via the sum of the SFS) derived by Gunnarsson \textit{et al.}~\cite{gunnarssonExactSiteFrequency2021a}, as found in Propositions~\ref{umbeta1prop} and~\ref{umprop} of the \hyperlink{si}{Supplementary Information}. Our stochastic time first-order approximation in equation~\eqref{birthdeathsfsapprox} matches theirs from the stochastic population scenario with a fixed elapsed time in the large-population limit, where the regimes coincide according to their convergence analysis~\cite{gunnarssonExactSiteFrequency2021a}. Our derivation for the pure-birth DD in equation~\eqref{purebirthdd} recovers a result from previous work on phylogenetic trees produced by Yule processes: combinatorics results relating to binary search trees~\cite{lynchMoreCombinatorialProperties1965a} were then applied to the phylogenetic context~\cite{steelPropertiesPhylogeneticTrees2001a}. 

Intuitively, we would think that the explicit single-cell MBD results from both the DD and the extra stochasticity arising from the mutational distribution at each past cell division (the internal nodes of the cell lineage tree). Surprisingly, we found that the the latter nodal stochasticity does not play a large role in the MBD. While there is certainly higher variance in the MBD than in the DD, as evidenced by Figures~\ref{mbd-fig}b--c, the shapes of the two distributions remain similar and we can construct the MBD based on the DD and $\mu$, the mean value of number of mutations acquired per cell in each past cell division. The derivation from equation~\eqref{poissonificationexpectationseq1} to~\eqref{poissonificationexpectationseq3} demonstrates that only the mean of the mutational distribution matters when obtaining the MBD, rather than its higher moments. We further tested this conclusion by applying other mutational distributions than the Poisson in stochastic simulations, which lead to the same predicted MBD, as shown in Figure~\ref{unif-geom-fig}. Employing the binning procedure described in equation~\eqref{histogrammingexpectations} of the \hyperlink{si}{Supplementary Information} allows us to retrace our steps from the MBD to the DD, which reinforces that it is only the mean of the mutational distribution that is of critical importance to the shape of the MBD, not the exact form of the distribution.
By considering the variances of the two distributions, we note that the variance in the single-cell MBD itself is growing while that of the mutational distribution is fixed. We thus expect that after sufficient events, the former will dominate.

We showed that the expected mutational burden for an arbitrary cell in a population (the mean of the MBD) increases logarithmically with the step count $i$ in our model (see Propositions~\ref{umbeta1prop} and~\ref{umprop}). In Moeller \& Mon P\`ere  \textit{et al.}'s continuous-time framework, this mean is shown to be the product of the expected number of divisions in the cell's past and the mutational mean $\mu$~\cite{moellerMeasuresGeneticDiversification2022}, much as we have argued in Figure~\ref{mbd-fig} for our conversion from the DD to the MBD. Under their intuitive assumption of mutation burdens arising from a compound Poisson distribution, the variance of the MBD is dependent on the means and the variances of the DD and the mutational (Poisson) distribution~\cite{moellerMeasuresGeneticDiversification2022}, whereas our derivations and simulations show that only the mean of the mutational distribution plays a significant role, not its higher moments.

Knowledge of the connection between the DD and the MBD also provides a means of evaluating the divisions in a cell's history. By reversing the argument in Figure~\ref{mbd-fig}, MBD data can provide the distribution of divisional histories in a cell population, without resorting direct measurements (for example, via telomere shortening~\cite{blascoTelomereLengthStem2007}).

While single-cell sequencing is still in its adolescence, grappling with hurdles such as trade-offs between sequencing noise, sample size and cost~\cite{goldmanImpactHeterogeneitySingleCell2019a, limAdvancingCancerResearch2020a}, there is a growing need and theoretical gap for mathematical and computational machinery to handle the vast quantities of data being produced~\cite{shapiroSinglecellSequencingbasedTechnologies2013a, choComparisonCellState2022}. Our model serves as a new framework to integrate single-cell and bulk information, and shows how various distributions of accumulated mutations are linked through the same stochastic process.

\subsection*{Acknowledgements}

 We thank Sabin Lessard, Nathaniel Mon P\`ere and Alexander Stein for fruitful discussions. This research was supported by the European Union's Horizon 2020 research and innovation programme under the Marie Sk\l odowska-Curie EvoGamesPlus grant number 955708.

\subsection*{Competing interests}

The authors declare no competing interests.

\printbibliography

@article{blascoTelomereLengthStem2007,
    author = {Blasco, Maria A.},
    title = {Telomere length, stem cells and aging},
    journal = {Nature Chemical Biology},
    year = {2007},
    month = {10},
    day = {01},
    volume = {3},
    number = {10},
    pages = {640-649},
    IGNOREabstract = {Telomere shortening occurs concomitant with organismal aging, and it is accelerated in the context of human diseases associated with mutations in telomerase, such as some cases of dyskeratosis congenita, idiopathic pulmonary fibrosis and aplastic anemia. People with these diseases, as well as Terc-deficient mice, show decreased lifespan coincidental with a premature loss of tissue renewal, which suggests that telomerase is rate-limiting for tissue homeostasis and organismal survival. These findings have gained special relevance as they suggest that telomerase activity and telomere length can directly affect the ability of stem cells to regenerate tissues. If this is true, stem cell dysfunction provoked by telomere shortening may be one of the mechanisms responsible for organismal aging in both humans and mice. Here, we will review the current evidence linking telomere shortening to aging and stem cell dysfunction.},
    issn = {1552-4469},
    doi = {10.1038/nchembio.2007.38},
    IGNOREurl = {https://doi.org/10.1038/nchembio.2007.38}
}

@book{pageMolecularEvolutionPhylogenetic2009,
  title = {Molecular Evolution: A Phylogenetic Approach},
  author = {Page, Roderick DM and Holmes, Edward C},
  year = {2009},
  publisher = {John Wiley \& Sons},
  chapter = {2}
}

@article{derenyiHowMutationAccumulation2024,
  title = {How mutation accumulation depends on the structure of the cell lineage tree},
  author = {Der\'enyi, Imre and Demeter, M\'arton C. and P\'erez-Jim\'enez, Mario and Grajzel, D\'aniel and Sz\"oll\ifmmode \mbox{\H{o}}\else \H{o}\fi{}si, Gergely J.},
  journal = {Phys. Rev. E},
  volume = {109},
  issue = {4},
  pages = {044407},
  numpages = {7},
  year = {2024},
  month = {4},
  publisher = {American Physical Society},
  doi = {10.1103/PhysRevE.109.044407},
  IGNOREurl = {https://link.aps.org/doi/10.1103/PhysRevE.109.044407}
}

@article{lee-sixPopulationDynamicsNormal2018a,
	title = {Population dynamics of normal human blood inferred from somatic mutations},
	volume = {561},
	IGNOREissn = {1476-4687},
	IGNOREurl = {https://doi.org/10.1038/s41586-018-0497-0},
	doi = {10.1038/s41586-018-0497-0},
	number = {7724},
	journal = {Nature},
	author = {Lee-Six, Henry and Øbro, Nina Friesgaard and Shepherd, Mairi S. and Grossmann, Sebastian and Dawson, Kevin and Belmonte, Miriam and Osborne, Robert J. and Huntly, Brian J. P. and Martincorena, Inigo and Anderson, Elizabeth and O’Neill, Laura and Stratton, Michael R. and Laurenti, Elisa and Green, Anthony R. and Kent, David G. and Campbell, Peter J.},
	IGNOREmonth = sep,
	year = {2018},
	pages = {473--478},
}

@article{wernerMeasuringSingleCell2020a,
	title = {Measuring single cell divisions in human tissues from multi-region sequencing data},
	volume = {11},
	IGNOREissn = {2041-1723},
	IGNOREurl = {https://doi.org/10.1038/s41467-020-14844-6},
	doi = {10.1038/s41467-020-14844-6},
	number = {1},
	journal = {Nature Communications},
	author = {Werner, Benjamin and Case, Jack and Williams, Marc J. and Chkhaidze, Ketevan and Temko, Daniel and Fernández-Mateos, Javier and Cresswell, George D. and Nichol, Daniel and Cross, William and Spiteri, Inmaculada and Huang, Weini and Tomlinson, Ian P. M. and Barnes, Chris P. and Graham, Trevor A. and Sottoriva, Andrea},
	IGNOREmonth = feb,
	year = {2020},
	pages = {1035},
}

@article{baeDifferentMutationalRates2018a,
	title = {Different mutational rates and mechanisms in human cells at pregastrulation and neurogenesis},
	volume = {359},
	IGNOREurl = {https://doi.org/10.1126/science.aan8690},
	doi = {10.1126/science.aan8690},
	number = {6375},
	IGNOREurldate = {2023-08-16},
	journal = {Science},
	author = {Bae, Taejeong and Tomasini, Livia and Mariani, Jessica and Zhou, Bo and Roychowdhury, Tanmoy and Franjic, Daniel and Pletikos, Mihovil and Pattni, Reenal and Chen, Bo-Juen and Venturini, Elisa and Riley-Gillis, Bridget and Sestan, Nenad and Urban, Alexander E. and Abyzov, Alexej and Vaccarino, Flora M.},
	IGNOREmonth = feb,
	year = {2018},
	note = {Publisher: American Association for the Advancement of Science},
	pages = {550--555},
}

@article{frankProblemsSomaticMutation2004a,
	title = {Problems of somatic mutation and cancer},
	volume = {26},
	IGNOREissn = {0265-9247},
	IGNOREurl = {https://doi.org/10.1002/bies.20000},
	doi = {10.1002/bies.20000},
	number = {3},
	IGNOREurldate = {2023-08-16},
	journal = {BioEssays},
	author = {Frank, Steven A. and Nowak, Martin A.},
	IGNOREmonth = mar,
	year = {2004},
	note = {Publisher: John Wiley \& Sons, Ltd},
	pages = {291--299},
}

@article{burrellCausesConsequencesGenetic2013a,
	title = {The causes and consequences of genetic heterogeneity in cancer evolution},
	volume = {501},
	IGNOREissn = {1476-4687},
	IGNOREurl = {https://doi.org/10.1038/nature12625},
	doi = {10.1038/nature12625},
	number = {7467},
	journal = {Nature},
	author = {Burrell, Rebecca A. and McGranahan, Nicholas and Bartek, Jiri and Swanton, Charles},
	IGNOREmonth = sep,
	year = {2013},
	pages = {338--345},
}

@article{komarovaCancerAgingOptimal2005a,
	title = {Cancer, aging and the optimal tissue design},
	volume = {15},
	IGNOREissn = {1044-579X},
	IGNOREurl = {https://www.sciencedirect.com/science/article/pii/S1044579X05000453},
	doi = {10.1016/j.semcancer.2005.07.003},
	number = {6},
	journal = {Somatic Evolution of Cancer Cells},
	author = {Komarova, Natalia L.},
	IGNOREmonth = dec,
	year = {2005},
	pages = {494--505},
}

@article{bozicAccumulationDriverPassenger2010a,
	title = {Accumulation of driver and passenger mutations during tumor progression},
	volume = {107},
	IGNOREurl = {https://doi.org/10.1073/pnas.1010978107},
	doi = {10.1073/pnas.1010978107},
	number = {43},
	IGNOREurldate = {2023-08-16},
	journal = {Proceedings of the National Academy of Sciences},
	author = {Bozic, Ivana and Antal, Tibor and Ohtsuki, Hisashi and Carter, Hannah and Kim, Dewey and Chen, Sining and Karchin, Rachel and Kinzler, Kenneth W. and Vogelstein, Bert and Nowak, Martin A.},
	IGNOREmonth = oct,
	year = {2010},
	note = {Publisher: Proceedings of the National Academy of Sciences},
	pages = {18545--18550},
}

@article{williamsIdentificationNeutralTumor2016a,
	title = {Identification of neutral tumor evolution across cancer types},
	volume = {48},
	IGNOREissn = {1546-1718},
	IGNOREurl = {https://doi.org/10.1038/ng.3489},
	doi = {10.1038/ng.3489},
	number = {3},
	journal = {Nature Genetics},
	author = {Williams, Marc J and Werner, Benjamin and Barnes, Chris P and Graham, Trevor A and Sottoriva, Andrea},
	IGNOREmonth = mar,
	year = {2016},
	pages = {238--244},
}

@article{lingExtremelyHighGenetic2015a,
	title = {Extremely high genetic diversity in a single tumor points to prevalence of non-{Darwinian} cell evolution},
	volume = {112},
	IGNOREurl = {https://doi.org/10.1073/pnas.1519556112},
	doi = {10.1073/pnas.1519556112},
	number = {47},
	IGNOREurldate = {2023-08-16},
	journal = {Proceedings of the National Academy of Sciences},
	author = {Ling, Shaoping and Hu, Zheng and Yang, Zuyu and Yang, Fang and Li, Yawei and Lin, Pei and Chen, Ke and Dong, Lili and Cao, Lihua and Tao, Yong and Hao, Lingtong and Chen, Qingjian and Gong, Qiang and Wu, Dafei and Li, Wenjie and Zhao, Wenming and Tian, Xiuyun and Hao, Chunyi and Hungate, Eric A. and Catenacci, Daniel V. T. and Hudson, Richard R. and Li, Wen-Hsiung and Lu, Xuemei and Wu, Chung-I},
	IGNOREmonth = nov,
	year = {2015},
	note = {Publisher: Proceedings of the National Academy of Sciences},
	pages = {E6496--E6505},
}

@article{sottorivaBigBangModel2015a,
	title = {A {Big} {Bang} model of human colorectal tumor growth},
	volume = {47},
	IGNOREissn = {1546-1718},
	IGNOREurl = {https://doi.org/10.1038/ng.3214},
	doi = {10.1038/ng.3214},
	number = {3},
	journal = {Nature Genetics},
	author = {Sottoriva, Andrea and Kang, Haeyoun and Ma, Zhicheng and Graham, Trevor A and Salomon, Matthew P and Zhao, Junsong and Marjoram, Paul and Siegmund, Kimberly and Press, Michael F and Shibata, Darryl and Curtis, Christina},
	IGNOREmonth = mar,
	year = {2015},
	pages = {209--216},
}

@article{kimuraGeneticVariabilityMaintained1968a,
	title = {Genetic variability maintained in a finite population due to mutational production of neutral and nearly neutral isoalleles},
	volume = {11},
	IGNOREissn = {0016-6723},
	IGNOREurl = {https://www.cambridge.org/core/article/genetic-variability-maintained-in-a-finite-population-due-to-mutational-production-of-neutral-and-nearly-neutral-isoalleles/A74BD3A5D72ED2C52444FD99DFE483EF},
	doi = {10.1017/S0016672300011459},
	number = {3},
	journal = {Genetics Research},
	author = {Kimura, Motoo},
	year = {1968},
	note = {Edition: 2009/04/14
Publisher: Cambridge University Press},
	pages = {247--270},
}

@article{ewensPseudotransientDistributionIts1964a,
	title = {The pseudo-transient distribution and its uses in genetics},
	volume = {1},
	IGNOREissn = {0021-9002},
	IGNOREurl = {https://www.cambridge.org/core/article/pseudotransient-distribution-and-its-uses-in-genetics/6B4F73A2EC2650F139D8CD3A7AC946FB},
	doi = {10.2307/3212065},
	number = {1},
	journal = {Journal of Applied Probability},
	author = {Ewens, W. J.},
	year = {1964},
	note = {Edition: 2016/07/14
Publisher: Cambridge University Press},
	pages = {141--156},
}

@article{fuStatisticalPropertiesSegregating1995a,
	title = {Statistical {Properties} of {Segregating} {Sites}},
	volume = {48},
	IGNOREissn = {0040-5809},
	IGNOREurl = {https://www.sciencedirect.com/science/article/pii/S0040580985710258},
	doi = {10.1006/tpbi.1995.1025},
	number = {2},
	journal = {Theoretical Population Biology},
	author = {Fu, Y.X.},
	IGNOREmonth = oct,
	year = {1995},
	pages = {172--197},
}

@article{bozicQuantifyingClonalSubclonal2016a,
	title = {Quantifying {Clonal} and {Subclonal} {Passenger} {Mutations} in {Cancer} {Evolution}},
	volume = {12},
	IGNOREurl = {https://doi.org/10.1371/journal.pcbi.1004731},
	doi = {10.1371/journal.pcbi.1004731},
	number = {2},
	journal = {PLOS Computational Biology},
	author = {Bozic, Ivana and Gerold, Jeffrey M. and Nowak, Martin A.},
	IGNOREmonth = feb,
	year = {2016},
	note = {Publisher: Public Library of Science},
	pages = {e1004731},
}

@article{durrettPopulationGeneticsNeutral2013,
	title = {Population {Genetics} of {Neutral} {Mutations} in {Exponentially} {Growing} {Cancer} {Cell} {Populations}},
	volume = {23},
	IGNOREissn = {10505164},
	IGNOREurl = {http://www.jstor.org/stable/23473207},
	number = {1},
	IGNOREurldate = {2023-08-16},
	journal = {The Annals of Applied Probability},
	author = {Durrett, Richard},
	year = {2013},
	note = {Publisher: Institute of Mathematical Statistics},
	pages = {230--250},
    doi = {10.1214/11-aap824}
}

@incollection{durrettBranchingProcessModels2015a,
	address = {Cham},
	title = {Branching {Process} {Models} of {Cancer}},
	IGNOREisbn = {978-3-319-16065-8},
	IGNOREurl = {https://doi.org/10.1007/978-3-319-16065-8_1},
	booktitle = {Branching {Process} {Models} of {Cancer}},
	publisher = {Springer International Publishing},
	author = {Durrett, Richard},
	editor = {Durrett, Richard},
	year = {2015},
	doi = {10.1007/978-3-319-16065-8_1},
	pages = {1--63},
}

@article{gunnarssonExactSiteFrequency2021a,
	title = {Exact site frequency spectra of neutrally evolving tumors: {A} transition between power laws reveals a signature of cell viability},
	volume = {142},
	IGNOREissn = {0040-5809},
	IGNOREurl = {https://www.sciencedirect.com/science/article/pii/S0040580921000666},
	doi = {10.1016/j.tpb.2021.09.004},
	journal = {Theoretical Population Biology},
	author = {Gunnarsson, Einar Bjarki and Leder, Kevin and Foo, Jasmine},
	IGNOREmonth = dec,
	year = {2021},
	pages = {67--90},
}

@article{shapiroSinglecellSequencingbasedTechnologies2013a,
	title = {Single-cell sequencing-based technologies will revolutionize whole-organism science},
	volume = {14},
	IGNOREissn = {1471-0064},
	IGNOREurl = {https://doi.org/10.1038/nrg3542},
	doi = {10.1038/nrg3542},
	number = {9},
	journal = {Nature Reviews Genetics},
	author = {Shapiro, Ehud and Biezuner, Tamir and Linnarsson, Sten},
	IGNOREmonth = sep,
	year = {2013},
	pages = {618--630},
}

@article{wangAdvancesApplicationsSingleCell2015a,
	title = {Advances and {Applications} of {Single}-{Cell} {Sequencing} {Technologies}},
	volume = {58},
	IGNOREissn = {1097-2765},
	IGNOREurl = {https://www.sciencedirect.com/science/article/pii/S109727651500341X},
	doi = {10.1016/j.molcel.2015.05.005},
	number = {4},
	journal = {Molecular Cell},
	author = {Wang, Yong and Navin, Nicholas E.},
	IGNOREmonth = may,
	year = {2015},
	pages = {598--609},
}

@article{chalmersAnalysis1000002017a,
	title = {Analysis of 100,000 human cancer genomes reveals the landscape of tumor mutational burden},
	volume = {9},
	IGNOREissn = {1756-994X},
	IGNOREurl = {https://doi.org/10.1186/s13073-017-0424-2},
	doi = {10.1186/s13073-017-0424-2},
	number = {1},
	journal = {Genome Medicine},
	author = {Chalmers, Zachary R. and Connelly, Caitlin F. and Fabrizio, David and Gay, Laurie and Ali, Siraj M. and Ennis, Riley and Schrock, Alexa and Campbell, Brittany and Shlien, Adam and Chmielecki, Juliann and Huang, Franklin and He, Yuting and Sun, James and Tabori, Uri and Kennedy, Mark and Lieber, Daniel S. and Roels, Steven and White, Jared and Otto, Geoffrey A. and Ross, Jeffrey S. and Garraway, Levi and Miller, Vincent A. and Stephens, Phillip J. and Frampton, Garrett M.},
	IGNOREmonth = apr,
	year = {2017},
	pages = {34},
}

@article{fernandezCancerSpecificThresholdsAdjust2019a,
	title = {Cancer-{Specific} {Thresholds} {Adjust} for {Whole} {Exome} {Sequencing}–{Based} {Tumor} {Mutational} {Burden} {Distribution}},
	IGNOREurl = {https://doi.org/10.1200/PO.18.00400},
	doi = {10.1200/PO.18.00400},
	number = {3},
	IGNOREurldate = {2023-08-16},
	journal = {JCO Precision Oncology},
	author = {Fernandez, Evan M. and Eng, Kenneth and Beg, Shaham and Beltran, Himisha and Faltas, Bishoy M. and Mosquera, Juan Miguel and Nanus, David M. and Pisapia, David J. and Rao, Rema A. and Robinson, Brian D. and Rubin, Mark A. and Elemento, Olivier and Sboner, Andrea and Shah, Manish A. and Song, Wei},
	IGNOREmonth = dec,
	year = {2019},
	note = {Publisher: Wolters Kluwer},
	pages = {1--12},
}

@article{martinez-perezPanelsModelsAccurate2021a,
	title = {Panels and models for accurate prediction of tumor mutation burden in tumor samples},
	volume = {5},
	IGNOREissn = {2397-768X},
	IGNOREurl = {https://doi.org/10.1038/s41698-021-00169-0},
	doi = {10.1038/s41698-021-00169-0},
	number = {1},
	journal = {npj Precision Oncology},
	author = {Martínez-Pérez, Elizabeth and Molina-Vila, Miguel Angel and Marino-Buslje, Cristina},
	IGNOREmonth = apr,
	year = {2021},
	pages = {31},
}

@article{steelPropertiesPhylogeneticTrees2001a,
	title = {Properties of phylogenetic trees generated by {Yule}-type speciation models},
	volume = {170},
	IGNOREissn = {0025-5564},
	IGNOREurl = {https://www.sciencedirect.com/science/article/pii/S0025556400000614},
	doi = {10.1016/S0025-5564(00)00061-4},
	number = {1},
	journal = {Mathematical Biosciences},
	author = {Steel, Mike and McKenzie, Andy},
	IGNOREmonth = mar,
	year = {2001},
	pages = {91--112},
}

@article{lynchMoreCombinatorialProperties1965a,
	title = {More {Combinatorial} {Properties} of {Certain} {Trees}},
	volume = {7},
	IGNOREissn = {0010-4620},
	IGNOREurl = {https://doi.org/10.1093/comjnl/7.4.299},
	doi = {10.1093/comjnl/7.4.299},
	number = {4},
	IGNOREurldate = {2023-08-16},
	journal = {The Computer Journal},
	author = {Lynch, William C.},
	IGNOREmonth = jan,
	year = {1965},
	pages = {299--302},
}

@article{kimuraNumberHeterozygousNucleotide1969a,
	title = {The number of heterozygous nucleotide sites maintained in a finite population due to steady flux of mutations},
	volume = {61},
	IGNOREissn = {1943-2631},
	IGNOREurl = {https://doi.org/10.1093/genetics/61.4.893},
	doi = {10.1093/genetics/61.4.893},
	number = {4},
	IGNOREurldate = {2023-08-16},
	journal = {Genetics},
	author = {Kimura, Motoo},
	IGNOREmonth = apr,
	year = {1969},
	pages = {893--903},
}

@article{zhengMathematicalIssuesArising2003a,
	title = {Mathematical {Issues} {Arising} {From} the {Directed} {Mutation} {Controversy}},
	volume = {164},
	IGNOREissn = {1943-2631},
	IGNOREurl = {https://doi.org/10.1093/genetics/164.1.373},
	doi = {10.1093/genetics/164.1.373},
	number = {1},
	IGNOREurldate = {2023-08-16},
	journal = {Genetics},
	author = {Zheng, Qi},
	IGNOREmonth = may,
	year = {2003},
	pages = {373--379},
}

@article{lederbergReplicaPlatingIndirect1989a,
	title = {Replica plating and indirect selection of bacterial mutants: isolation of preadaptive mutants in bacteria by sib selection.},
	volume = {121},
	IGNOREissn = {1943-2631},
	IGNOREurl = {https://doi.org/10.1093/genetics/121.3.395},
	doi = {10.1093/genetics/121.3.395},
	number = {3},
	IGNOREurldate = {2023-08-16},
	journal = {Genetics},
	author = {Lederberg, J},
	IGNOREmonth = mar,
	year = {1989},
	pages = {395--399},
}

@article{ronenLearningNaturalSelection2013a,
	title = {Learning {Natural} {Selection} from the {Site} {Frequency} {Spectrum}},
	volume = {195},
	IGNOREissn = {1943-2631},
	IGNOREurl = {https://doi.org/10.1534/genetics.113.152587},
	doi = {10.1534/genetics.113.152587},
	number = {1},
	IGNOREurldate = {2023-08-16},
	journal = {Genetics},
	author = {Ronen, Roy and Udpa, Nitin and Halperin, Eran and Bafna, Vineet},
	IGNOREmonth = sep,
	year = {2013},
	pages = {181--193},
}

@article{poonSynonymousMutationsReveal2021a,
	title = {Synonymous mutations reveal genome-wide levels of positive selection in healthy tissues},
	volume = {53},
	IGNOREissn = {1546-1718},
	IGNOREurl = {https://doi.org/10.1038/s41588-021-00957-1},
	doi = {10.1038/s41588-021-00957-1},
	number = {11},
	journal = {Nature Genetics},
	author = {Poon, Gladys Y. P. and Watson, Caroline J. and Fisher, Daniel S. and Blundell, Jamie R.},
	IGNOREmonth = nov,
	year = {2021},
	pages = {1597--1605},
}

@article{loebExtensiveSubclonalMutational2019a,
	title = {Extensive subclonal mutational diversity in human colorectal cancer and its significance},
	volume = {116},
	IGNOREurl = {https://doi.org/10.1073/pnas.1910301116},
	doi = {10.1073/pnas.1910301116},
	number = {52},
	IGNOREurldate = {2023-08-16},
	journal = {Proceedings of the National Academy of Sciences},
	author = {Loeb, Lawrence A. and Kohrn, Brendan F. and Loubet-Senear, Kaitlyn J. and Dunn, Yasmin J. and Ahn, Eun Hyun and O’Sullivan, Jacintha N. and Salk, Jesse J. and Bronner, Mary P. and Beckman, Robert A.},
	IGNOREmonth = dec,
	year = {2019},
	note = {Publisher: Proceedings of the National Academy of Sciences},
	pages = {26863--26872},
}

@article{simonsDeepSequencingProbe2016a,
	title = {Deep sequencing as a probe of normal stem cell fate and preneoplasia in human epidermis},
	volume = {113},
	IGNOREurl = {https://doi.org/10.1073/pnas.1516123113},
	doi = {10.1073/pnas.1516123113},
	number = {1},
	IGNOREurldate = {2023-08-16},
	journal = {Proceedings of the National Academy of Sciences},
	author = {Simons, Benjamin D.},
	IGNOREmonth = jan,
	year = {2016},
	note = {Publisher: Proceedings of the National Academy of Sciences},
	pages = {128--133},
}

@article{kurpasModesSelectionTumors2022a,
	title = {Modes of {Selection} in {Tumors} as {Reflected} by {Two} {Mathematical} {Models} and {Site} {Frequency} {Spectra}},
	volume = {10},
	IGNOREissn = {2296-701X},
	IGNOREurl = {https://www.frontiersin.org/articles/10.3389/fevo.2022.889438},
	journal = {Frontiers in Ecology and Evolution},
	author = {Kurpas, Monika K. and Kimmel, Marek},
	year = {2022},
    doi = {10.3389/fevo.2022.889438}
}

@article{watsonEvolutionaryDynamicsFitness2020a,
	title = {The evolutionary dynamics and fitness landscape of clonal hematopoiesis},
	volume = {367},
	IGNOREurl = {https://doi.org/10.1126/science.aay9333},
	doi = {10.1126/science.aay9333},
	number = {6485},
	IGNOREurldate = {2023-08-16},
	journal = {Science},
	author = {Watson, Caroline J. and Papula, A. L. and Poon, Gladys Y. P. and Wong, Wing H. and Young, Andrew L. and Druley, Todd E. and Fisher, Daniel S. and Blundell, Jamie R.},
	IGNOREmonth = mar,
	year = {2020},
	note = {Publisher: American Association for the Advancement of Science},
	pages = {1449--1454},
}

@article{choComparisonCellState2022,
	title = {Comparison of cell state models derived from single-cell {RNA} sequencing data: graph versus multi-dimensional space},
	volume = {19},
	IGNOREissn = {1551-0018},
	IGNOREurl = {https://www.aimspress.com/article/doi/10.3934/mbe.2022395},
	doi = {10.3934/mbe.2022395},
	number = {8},
	journal = {Mathematical Biosciences and Engineering},
	author = {Cho, Heyrim and Kuo, Ya-Huei and Rockne, Russell C.},
	year = {2022},
	pages = {8505--8536},
}

@article{limAdvancingCancerResearch2020a,
	title = {Advancing {Cancer} {Research} and {Medicine} with {Single}-{Cell} {Genomics}},
	volume = {37},
	IGNOREissn = {1535-6108},
	IGNOREurl = {https://www.sciencedirect.com/science/article/pii/S1535610820301483},
	doi = {10.1016/j.ccell.2020.03.008},
	number = {4},
	journal = {Cancer Cell},
	author = {Lim, Bora and Lin, Yiyun and Navin, Nicholas},
	IGNOREmonth = apr,
	year = {2020},
	pages = {456--470},
}

@article{goldmanImpactHeterogeneitySingleCell2019a,
	title = {The {Impact} of {Heterogeneity} on {Single}-{Cell} {Sequencing}},
	volume = {10},
	IGNOREissn = {1664-8021},
	IGNOREurl = {https://www.frontiersin.org/articles/10.3389/fgene.2019.00008},
	journal = {Frontiers in Genetics},
	author = {Goldman, Samantha L. and MacKay, Matthew and Afshinnekoo, Ebrahim and Melnick, Ari M. and Wu, Shuxiu and Mason, Christopher E.},
	year = {2019},
    doi = {10.3389/fgene.2019.00008}
}

@article{deijfenGrowingNetworksPreferential2009a,
	title = {Growing networks with preferential deletion and addition of edges},
	volume = {388},
	IGNOREissn = {0378-4371},
	IGNOREurl = {https://www.sciencedirect.com/science/article/pii/S0378437109004890},
	doi = {10.1016/j.physa.2009.06.032},
	number = {19},
	journal = {Physica A: Statistical Mechanics and its Applications},
	author = {Deijfen, Maria and Lindholm, Mathias},
	IGNOREmonth = oct,
	year = {2009},
	pages = {4297--4303},
}

@article{abascalSomaticMutationLandscapes2021a,
	title = {Somatic mutation landscapes at single-molecule resolution},
	volume = {593},
	IGNOREissn = {1476-4687},
	IGNOREurl = {https://doi.org/10.1038/s41586-021-03477-4},
	doi = {10.1038/s41586-021-03477-4},
	number = {7859},
	journal = {Nature},
	author = {Abascal, Federico and Harvey, Luke M. R. and Mitchell, Emily and Lawson, Andrew R. J. and Lensing, Stefanie V. and Ellis, Peter and Russell, Andrew J. C. and Alcantara, Raul E. and Baez-Ortega, Adrian and Wang, Yichen and Kwa, Eugene Jing and Lee-Six, Henry and Cagan, Alex and Coorens, Tim H. H. and Chapman, Michael Spencer and Olafsson, Sigurgeir and Leonard, Steven and Jones, David and Machado, Heather E. and Davies, Megan and Øbro, Nina F. and Mahubani, Krishnaa T. and Allinson, Kieren and Gerstung, Moritz and Saeb-Parsy, Kourosh and Kent, David G. and Laurenti, Elisa and Stratton, Michael R. and Rahbari, Raheleh and Campbell, Peter J. and Osborne, Robert J. and Martincorena, Iñigo},
	IGNOREmonth = may,
	year = {2021},
	pages = {405--410},
}

@book{weinberg2013biology,
  title={The Biology of Cancer},
  author={Weinberg, Robert A},
  year={2013},
  publisher={Garland Science},
  chapter={1},
  doi = {10.1201/9780203852569}
}

@book{fellerIntroductionProbabilityTheory1968,
  title={An Introduction to Probability Theory and Its Applications},
  author={Feller, William},
  volume={1},
  year={1968},
  publisher={John Wiley \& Sons},
  chapter={XIV}
}

@article{moellerMeasuresGeneticDiversification2022,
	doi = {10.7554/elife.89780.1},
	IGNOREurl = {https://doi.org/10.7554/elife.89780.1},
	year = 2023,
	month = {9},
	publisher = {{eLife} Sciences Publications, Ltd},
	author = {Marius E. Moeller and Nathaniel V. Mon P\`ere and Benjamin Werner and Weini Huang},
	title = {Measures of genetic diversification in somatic tissues at bulk and single cell resolution suggest sources of unknown stochasticity}
}

@book{gradshteyn2007table,
  title={Table of Integrals, Series, and Products},
  author={Gradshteyn, Izrail Solomonovich and Ryzhik, Iosif Moiseevich},
  year={2007},
  edition={7},
  page={317},
  publisher={Academic Press}
}

@book{durrettProbabilityModelsDNA2008,
    author = {Durrett, Richard},
    title = {Probability Models for DNA Sequence Evolution},
    publisher = {Springer New York, NY},
    year = {2008},
    edition = {2},
    doi = {10.1007/978-0-387-78168-6},
    IGNOREissn = {1431-7028},
    series = {Probability and Its Applications}
}

@article{karlssonDeterministicEvolutionStringent2023a,
	title = {Deterministic evolution and stringent selection during preneoplasia},
	volume = {618},
	IGNOREissn = {1476-4687},
	IGNOREurl = {https://doi.org/10.1038/s41586-023-06102-8},
	doi = {10.1038/s41586-023-06102-8},
	number = {7964},
	journal = {Nature},
	author = {Karlsson, Kasper and Przybilla, Moritz J. and Kotler, Eran and Khan, Aziz and Xu, Hang and Karagyozova, Kremena and Sockell, Alexandra and Wong, Wing H. and Liu, Katherine and Mah, Amanda and Lo, Yuan-Hung and Lu, Bingxin and Houlahan, Kathleen E. and Ma, Zhicheng and Suarez, Carlos J. and Barnes, Chris P. and Kuo, Calvin J. and Curtis, Christina},
	IGNOREmonth = jun,
	year = {2023},
	pages = {383--393},
}

@article{parkClonalInterferenceLarge2007a,
	title = {Clonal interference in large populations},
	volume = {104},
	IGNOREurl = {https://doi.org/10.1073/pnas.0705778104},
	doi = {10.1073/pnas.0705778104},
	number = {46},
	IGNOREurldate = {2023-08-16},
	journal = {Proceedings of the National Academy of Sciences},
	author = {Park, Su-Chan and Krug, Joachim},
	IGNOREmonth = nov,
	year = {2007},
	note = {Publisher: Proceedings of the National Academy of Sciences},
	pages = {18135--18140},
}

@article{reuschEvolutionSomaticGenetic2021a,
	title = {Evolution via somatic genetic variation in modular species},
	volume = {36},
	IGNOREissn = {0169-5347},
	IGNOREurl = {https://doi.org/10.1016/j.tree.2021.08.011},
	doi = {10.1016/j.tree.2021.08.011},
	number = {12},
	IGNOREurldate = {2023-08-16},
	journal = {Trends in Ecology \& Evolution},
	author = {Reusch, Thorsten B.H. and Baums, Iliana B. and Werner, Benjamin},
	IGNOREmonth = dec,
	year = {2021},
	note = {Publisher: Elsevier},
	pages = {1083--1092},
}

@article{tungSignaturesNeutralEvolution2021a,
	title = {Signatures of neutral evolution in exponentially growing tumors: {A} theoretical perspective},
	volume = {17},
	IGNOREurl = {https://doi.org/10.1371/journal.pcbi.1008701},
	doi = {10.1371/journal.pcbi.1008701},
	number = {2},
	journal = {PLOS Computational Biology},
	author = {Tung, Hwai-Ray and Durrett, Richard},
	IGNOREmonth = feb,
	year = {2021},
	note = {Publisher: Public Library of Science},
	pages = {e1008701},
}

\newpage

\setcounter{equation}{0}
\setcounter{figure}{0}
\setcounter{table}{0}
\renewcommand{\theequation}{S\arabic{equation}}
\renewcommand{\thefigure}{S\arabic{figure}}
\renewcommand{\thetable}{S\arabic{table}}
\renewcommand{\theprop}{S\arabic{prop}}
\renewcommand{\thelem}{S\arabic{lem}}
\renewcommand{\thecoro}{S\arabic{coro}}

\section*{Supplementary Information}\hypertarget{si}{}

We first derive results about the expected population size described by the Markov chain depicted in Figure~\ref{chain-tree-fig}a and its probability of extinction, which allow us to make rigorous our main text statement that all of our expectations are conditioned on survival. We then prove several results from the main text. Next, we make additional statements about the mutational occurrences $C_i$ defined by~\eqref{sfsmbdmirror} and the number of unique mutations $M_i$ using our recurrence relation approach. Finally, we comment on what happens if other (than Poisson) mutational distributions are used, and we define a conversion from a single-cell mutational burden distribution (MBD) to a division distribution (DD). Table~\ref{sinotationtab} summarises the new notation for Supplementary Information. 

\begin{table}[ht]
    \centering
    \begin{tabular}{cl}
    \toprule
        \textbf{Symbol} & \textbf{Description}\\\hline
        $\gamma_z$ & Extinction probability for a birth-death process starting at $N_0 = z$\\
        $u_{z, n}(p, q; a)$ & Random walk absorption probability; see Proposition~\ref{fellertrigprop} for full definition\\
        $T$ & Markov chain transition matrix, with entries $T_{m, n} = \P{N_i = n \, \middle| \, N_{i - 1} = m}$\\
        $\zeta_{k, i}$ & coefficient equal to $1 - (i - k) / (k + 1)$ for $\lfloor(i + 1) / 2\rfloor \leq k \leq i$\\
        $C_i$ & Number of mutational occurrences at step $i$: the sum of the entries of $Y_i$\\\bottomrule
    \end{tabular}
    \caption{Notation used in Supplementary Information, in addition to those in Table~\ref{notationtab}.}
    \label{sinotationtab}
\end{table}

\subsection*{Extinction probabilities and expected population sizes}

Consider the Markov chain on a finite state space depicted by Figure~\ref{chain-tree-fig}a.
\begin{prop}\label{gammaprop}
    The extinction probability for the birth-death process starting at $N_0 = z$ with $\beta \neq \frac{1}{2}$ is
    \begin{equation}\label{gammadef}
        \P{N_i = 0 \text{ for some } i \geq 0 \, \middle| \, N_0 = z} = \gamma_z = \frac{\left(\delta / \beta\right)^N - \left(\delta / \beta\right)^z}{\left(\delta / \beta\right)^N - 1}.
    \end{equation}
\end{prop}
\begin{proof}
    Conditioning on the previous event (birth or death), we have for $1 \leq z \leq N - 1$:
    \begin{equation}\label{gammarecrel}
        \gamma_z = \beta\gamma_{z + 1} + \delta\gamma_{z - 1},
    \end{equation}
    with boundary conditions $\gamma_0 = 1$ and $\gamma_N = 0$. We note that~\eqref{gammarecrel} is a linear homogeneous second-order recurrence relation, so if we find two particular solutions (in this case $1$ and $\left(\delta / \beta\right)^z$), the general solution is a linear combination of them. The coefficients $A$ and $B$ of this linear combination are found via the boundary conditions:
    \begin{equation*}
        A + B = 1 \quad \text{ and } \quad A + B\left(\delta / \beta\right)^N = 0,
    \end{equation*}
    which solve to give~\eqref{gammadef}.
\end{proof}
\begin{rem}
    By symmetry, we note that the probability to absorb into the $N$ state (having started at $N_0 = z$) is equal to $1 - \gamma_z$, as can be seen by making the change of variables
    \begin{equation}\label{reversemarkov}
        \{\beta, \delta, z\} \to \{\delta, \beta, N - z\},
    \end{equation}
    which allows us to view the Markov chain in reverse, in a sense. Therefore we conclude that only on a set of zero probability does this Markov chain not reach an absorbing state. Sometimes, we will consider the large-population limit $N \to \infty$, where then the extinction probability is simply $(\delta / \beta)^z$. For large enough $i$, we then obtain the approximation for the survival probability $\P{N_i > 0 \, \middle| \, N_0 = 1} \simeq 1 - \delta / \beta$.
\end{rem}
\begin{rem}
    A second interpretation of this Markov chain is of a biased random walk on the integers $\{0, \dots, N\}$. This specific set-up with absorbing boundaries is known as the gambler's ruin problem; the change of variables in~\eqref{reversemarkov} is analogous to putting ourselves in the other gambler's shoes.
\end{rem}
We can now describe the expected value of this Markov chain at step $i$. We assume that $N$ is large enough; that is, that the maximum population size $i + 1$ (having started at $N_0 = 1$) is less than $N$.
\begin{prop}\label{expectedniprop}
    The expected population at step $i$ of the birth-death process in Figure~\ref{chain-tree-fig}a is
    \begin{equation}\label{noncondexpectedni}
        \E{N_i \, \middle| \, N_0 = 1} = \sum_{\substack{n = 1 \\ i - n \text{ odd}}}^{i + 1}\frac{n^2}{i + 1}\binom{i + 1}{\frac{i - (n - 1)}{2}}\beta^{\frac{i + (n - 1)}{2}}\delta^{\frac{i - (n - 1)}{2}},
    \end{equation}
    and when conditioned on survival, we have $\E{N_i \, \middle| \, N_0 = 1, N_i > 0} = \E{N_i \, \middle| \, N_0 = 1} / \P{N_i > 0 \, \middle| \, N_0 = 1}$.
\end{prop}
\begin{proof}
    Consider the transition matrix $T$ (independent of $i \geq 1$ by the homogeneity of our Markov chain):
    \begin{equation*}\label{transitionmatrix}
        T = \begin{pmatrix}
            1 & 0 & 0 & \cdots &  0 & 0\\
            \delta & 0 & \beta & \cdots & 0 & 0\\
            0 & \delta & 0 & \cdots & 0 & 0\\
            \vdots & \vdots & \vdots & \ddots & \vdots & \vdots\\
            0 & 0 & 0 & \cdots & 0 & \beta\\
            0 & 0 & 0 & \cdots & 0 & 1
        \end{pmatrix} \quad \text{ for entries } \quad T_{m, n} = \P{N_i = n \, \middle| \, N_{i - 1} = m}.
    \end{equation*}
    A property of stochastic Markov matrices is that the entries of the powers of the matrix represent the multi-step probabilities: $T^i_{m, n} = \P{N_i = n \, \middle| \, N_0 = m}$ (this can be simply proven, or concluded from the Chapman-Kolmogorov equations). We are interested in the first row of these matrices $T^i$: that is, the case $m = 1$. We will consider $N > i + 1$, so that the maximum attainable population in $i$ steps (having started at $N_0 = 1$) is still less than the maximal and absorbing state $N$.
    
    Note that the survival probability $\P{N_i > 0 \, \middle| \, N_0 = 1}$ is thus
    \begin{equation}\label{survivalprob}
        \P{N_i > 0 \, \middle| \, N_0 = 1} = \sum_{n = 1}^{i + 1} T^i_{1, n} = 1 - T^i_{1, 0}.
    \end{equation}
    This sum as $i \to \infty$ is $1 - \delta / \beta$, as noted in the remark following Proposition~\ref{gammaprop}.
    
    For $1 \leq i < N - 1$ the entries of the matrix powers satisfy the following recurrence relation
    \begin{equation}\label{ti1nrr}
        T^{i + 1}_{1, n} = T_{n - 1, n}T^i_{1, n - 1} + T_{n + 1, n}T^i_{1, n + 1} = \beta T^i_{1, n - 1} + \delta T^i_{1, n + 1} \quad \text{ for } \quad 1 < n < N - 1,
    \end{equation}
    along with $T^{i + 1}_{1, 0} = T^i_{1, 0} + \delta T^i_{1, 1}$, $T^{i + 1}_{1, 1} = \delta T^i_{1, 2}$, $T^{i + 1}_{1, N - 1} = \beta T^i_{1, N - 2}$ and $T^{i + 1}_{1, N} = \beta T^i_{1, N - 1} + T^i_{1, N}$.

    For $n > 0$, consider the transition probability $T^i_{1, n}$: undergoing births or deaths, $i$ steps are taken to progress from state $1$ to state $n$. So, the sum of powers of $\beta$ and $\delta$ should be $i$ and their difference should be $n - 1$. The integer coefficients turn out to form the Catalan triangle (see the proof of Corollary~\ref{bdseqcoro} for an interpretation of this fact), so we obtain the ansatz
    \begin{equation}\label{ti1n}
        T^i_{1, n} = \frac{n}{i + 1}\binom{i + 1}{\frac{i - (n - 1)}{2}}\beta^{\frac{i + (n - 1)}{2}}\delta^{\frac{i - (n - 1)}{2}} \quad \text{ for } \quad 1 \leq n \leq i + 1 < N,
    \end{equation}
    where the binomial coefficients are defined to be zero if $i$ and $n$ have the same parity, since it is impossible to reach an odd (respectively even) destination $n$ in an even (respectively odd) number of steps having started from $1$. It is straightforward to show that the expression~\eqref{ti1n} satisfies the recurrence~\eqref{ti1nrr}:
    \begin{align*}
        \beta T^i_{1, n - 1} + \delta T^i_{1, n + 1} &= \beta\frac{n - 1}{i + 1}\binom{i + 1}{\frac{i - (n - 2)}{2}}\beta^{\frac{i + (n - 2)}{2}}\delta^{\frac{i - (n - 2)}{2}} + \delta\frac{n + 1}{i + 1}\binom{i + 1}{\frac{i - n}{2}}\beta^{\frac{i + n}{2}}\delta^{\frac{i - n}{2}}\\
        &= \frac{i! \beta^{\frac{i + n}{2}}\delta^{\frac{i - n}{2} + 1}}{\left(\frac{i - n}{2}\right)!\left(\frac{i + n}{2}\right)!}\left(\frac{n - 1}{\frac{i - n}{2} + 1} + \frac{n + 1}{\frac{i + n}{2} + 1}\right)\\
        &= \frac{n(i + 1)!}{\left(\frac{i - n}{2} + 1\right)!\left(\frac{i + n}{2} + 1\right)!} \beta^{\frac{i + n}{2}}\delta^{\frac{i - n}{2} + 1}\\
        &= T^{i + 1}_{1, n}.
    \end{align*}
    Using~\eqref{ti1n}, the expected population size non-conditioned on survival is then
    \begin{equation*}
        \E{N_i \, \middle| \, N_0 = 1} = \sum_{n = 0}^{i + 1} n\P{N_i = n \, \middle| \, N_0 = 1} = \sum_{\substack{n = 1 \\ i - n \text{ odd}}}^{i + 1}\frac{n^2}{i + 1}\binom{i + 1}{\frac{i - (n - 1)}{2}}\beta^{\frac{i + (n - 1)}{2}}\delta^{\frac{i - (n - 1)}{2}},
    \end{equation*}
    as desired. Conditioning on survival, we find
    \begin{align*}
        \E{N_i \, \middle| \, N_0 = 1, N_i > 0} &= \sum_{n = 0}^{i + 1} n\P{N_i = n \, \middle| \, N_0 = 1, N_i > 0}\\
        &= \sum_{n = 0}^{i + 1} n\frac{\P{N_i = n \cap N_i > 0 \, \middle| \, N_0 = 1}}{\P{N_i > 0 \, \middle| \, N_0 = 1}}\\
        &= \sum_{n = 1}^{i + 1} n\frac{\P{N_i = n \, \middle| \, N_0 = 1}}{\P{N_i > 0 \, \middle| \, N_0 = 1}}\\
        &= \frac{\E{N_i \, \middle| \, N_0 = 1}}{\P{N_i > 0 \, \middle| \, N_0 = 1}}.\qedhere
    \end{align*}
\end{proof}
\begin{coro}\label{zetacoro}
    For $k$ the number of births out of $i \leq 2k$ trials, we have
    \begin{equation}\label{ti1nzetabinom}
        T^i_{1, k - (i - k) + 1} = \zeta_{k, i}\P{X = k},
    \end{equation}
    where $X \sim \text{Binom}(i, \beta)$ is a binomially-distributed random variable with probability $\beta$ of success and $\zeta_{k, i}$ is defined for $\lfloor(i + 1) / 2\rfloor \leq k \leq i$ by
    \begin{equation}\label{zetadef}
        \zeta_{k, i} = 1 - \frac{i - k}{k + 1} \in (0, 1].
    \end{equation}
\end{coro}
\begin{proof}
    Write $k$ for the number of births out of the $i$ trials, so that $n = k - (i - k) + 1$. Then,~\eqref{ti1n} becomes
    \begin{equation*}
        T^i_{1, n} = \frac{2n}{i + n + 1}\P{X = \frac{i + (n - 1)}{2}} = \left(1 - \frac{i - k}{k + 1}\right)\P{X = k}.\qedhere
    \end{equation*}
\end{proof}
\begin{rem}
    By itself, $\P{X = k}$ gives the probability of $k$ successes out of $i$ trials, where the successes can occur anywhere. The inclusion of a factor of $\zeta_{k, i}$ in~\eqref{ti1nzetabinom} thus describes the case where the number of successes (births) must outweigh the number of failures (deaths) after any of the first $1 \leq i' \leq i$ trials had the experiment ended there, since we require the population to be non-extinct at each step. For example, we want the ordered sequence $bbdd$ of births ($b$) and deaths ($d$) to contribute to $T^4_{1, 1}$, but we don't want to count the sequence $bddb$, as the final $b$ in the latter isn't sensible, since the population is extinct after the third step. To formalise this, we state the following definitions and ensuing corollary.
\end{rem}
A $(p, q)$-birth-death sequence is a word made from $p$ letters $b$ and $q$ letters $d$, and a non-empty word obtained by removing (possibly zero) letters from the end of the original word is called a truncation. For example, $bbdbd$ and $bbdb$ are both truncations of the $(3, 2)$-birth-death sequence $bbdbd$. A $(p, q)$-birth-death sequence is called surviving if none of its $p + q$ truncations contain strictly more $d$s than $b$s. We then have the following interpretation of the coefficients $\zeta_{k, i}$.
\begin{coro}\label{bdseqcoro}
    The proportion of surviving $(k, i - k)$-birth-death sequences is $\zeta_{k, i}$.
\end{coro}
\begin{proof}
    By observation, the binomial coefficient $\binom{i}{k}$ is the total number of $(k, i - k)$-birth-death sequences. On the other hand, another interpretation of the Catalan triangle coefficients $\zeta_{k, i}\binom{i}{k}$ arising in~\eqref{ti1n} is the number of $i$-step walks from $(0,1)$ to $(i, k - (i - k) + 1)$, where each step goes from $(x, y)$ to $(x + 1, y \pm 1)$ and the walk stays in the positive quadrant. Therefore, they count the number of surviving $(k, i - k)$-birth-death sequences.
\end{proof}
\begin{rem}    
    Indeed, the definition~\eqref{zetadef} of $\zeta_{k, i}$ resembles the large-$i$ survival probability $1 - \delta / \beta$, since $i - k$ is the number of deaths and $k + 1$ is one more than the number of births. Thus the expression~\eqref{ti1nzetabinom} can be interpreted as this `survival probability' $\zeta_{k, i}$ multiplied by a binomial distribution, which doesn't see the absorbing boundary at $0$.
\end{rem}
Now, we can use our new expression~\eqref{ti1nzetabinom} for the transition probabilities to compare the exact expected population with an intuitive linear approximation.
\begin{rem}
    The sums in the expected population~\eqref{noncondexpectedni} and the survival probability~\eqref{survivalprob} only involve terms $1 \leq n \leq i + 1$ where $i - n$ is odd. When $n = k - (i - k) + 1$, this translates to summing over integers $\lfloor(i + 1) / 2\rfloor \leq k \leq i$. Using~\eqref{survivalprob} and~\eqref{ti1nzetabinom}, the expected population size conditional on survival is then
    \begin{equation}\label{condexpectedni}
        \E{N_i \, \middle| \, N_0 = 1, N_i > 0} = \frac{\sum_{k = \lfloor\frac{i + 1}{2}\rfloor}^i(k - (i - k) + 1)\zeta_{k, i}\P{X = k}}{\sum_{k' = \lfloor\frac{i + 1}{2}\rfloor}^i\zeta_{k', i}\P{X = k'}}.
    \end{equation}
\end{rem}
\begin{rem}
    Consider a na\"ive derivation of the expected population size, conditioning on step $i - 1$
    \begin{equation}\label{linearapprox}
        N_i = \beta (N_{i - 1} + 1) + \delta(N_{i - 1} - 1) = N_{i - 1} + \beta - \delta = \cdots = (\beta - \delta)i + 1,
    \end{equation}
    and taking expectations of both sides. This is an appropriate linear approximation in the limit of low death, but doesn't correctly condition on survival, as shown by the discrepancy for low $\beta$ (see Figure~\ref{n-fig}).
\end{rem}
\begin{figure}[ht]
    \centering
    \includegraphics[width = 0.61\textwidth]{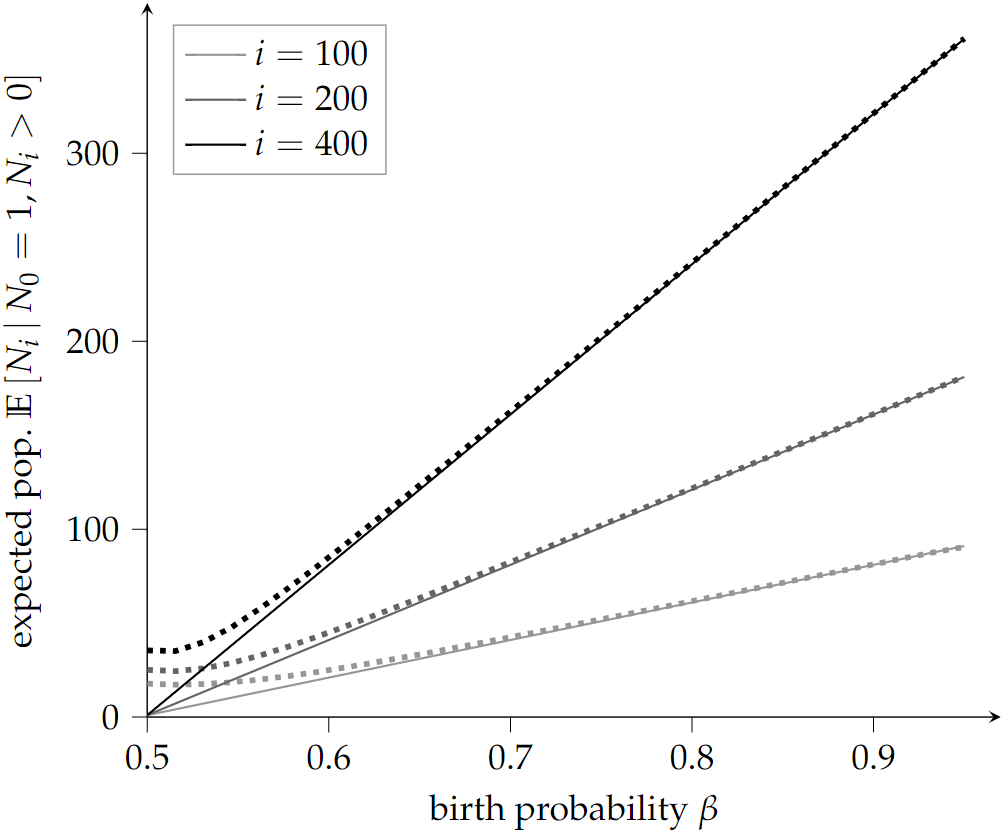}
    \caption{Plot of the expected population size conditioned on survival $\E{N_i \, \middle| \, N_0 = 1, N_i > 0}$ versus the birth probability $\beta$ for different numbers of steps $i = 100$ (pale grey), $i = 200$ (dark grey) and $i = 400$ (black). The dotted lines show the exact expression from Proposition~\ref{expectedniprop}, and the solid lines are the linear approximation $(\beta - \delta)i + 1$, showing excellent agreement for low death $\delta \ll \beta$ regardless of $i$.}
    \label{n-fig}
\end{figure}
\begin{rem}
     A heuristic derivation of the linear approximation~\eqref{linearapprox} from the exact expression~\eqref{condexpectedni} is the following. For large $i$, the probability mass function of the binomial distribution tends to a function with mean $i\beta = k$, the expected number of births in $i$ steps, and variance $i\beta(1 - \beta)$. When computing the sums in~\eqref{condexpectedni} for large $i$, consider the dominant $k = i\beta$ term (where we gloss over whether or not $i\beta$ is an integer) as the primary contribution to the sum:
     \begin{equation*}
        \frac{\sum_k(k - (i - k) + 1)\zeta_{k, i}\P{X = k}}{\sum_{k'}\zeta_{k', i}\P{X = k'}} \simeq \frac{\left(i\beta - (i - i\beta) + 1\right)\zeta_{i\beta, i}}{\zeta_{i\beta, i}} = (2\beta - 1)i + 1,
     \end{equation*}
     recovering~\eqref{linearapprox}. It is also worth noting that $\zeta_{i\beta, i} = 1 - (i - i\beta) / (i\beta + 1)$ coincides with the large-$i$ survival probability $1 - \delta / \beta$, supporting our intuition that $\zeta_{k, i}$ can be thought of as a survival probability of sorts, as anticipated by Corollary~\ref{bdseqcoro} and the remark that follows it.
\end{rem}
The exact survival probability at step $i$ is given explicitly by~\eqref{survivalprob} and~\eqref{ti1n}. A related result from Feller~\cite{fellerIntroductionProbabilityTheory1968} gives the distribution of times until the Markov process goes extinct.
\begin{prop}[Feller]\label{fellertrigprop}
    Consider a random walk on the integers $\{0, \dots, a\}$ (with absorption on both end points) beginning at $z$ with probability $p \in (0, 1)$ of moving right and $q = 1 - p$ of moving left. The probability of absorption on the $0$ boundary at step $n$ is
    \begin{equation}\label{felleroriginal}
        u_{z, n}(p, q; a) = 2^n p^{(n - z) / 2} q^{(n + z) / 2}\frac{1}{a}\sum_{k = 1}^{a - 1}\cos^{n - 1}\frac{\uppi k}{a}\sin\frac{\uppi k}{a}\sin\frac{\uppi zk}{a}.
    \end{equation}
\end{prop}
\begin{proof}
    See Chapter XIV,4 of~\cite{fellerIntroductionProbabilityTheory1968}.
\end{proof}
\begin{rem}
    Sending $a \to \infty$ in~\eqref{felleroriginal} provides a simplification~\cite{fellerIntroductionProbabilityTheory1968}:
    \begin{equation}\label{feller2}
        \lim_{a \to \infty} u_{z, n}(p, q; a) = \frac{z}{n}\binom{n}{\frac{n - z}{2}}p^{\frac{n - z}{2}}q^{\frac{n + z}{2}},
    \end{equation}
    where the binomial coefficient implicitly assumes that $n$ and $z$ have the same parity. This is the case of only a single absorbing boundary at $0$, as in our considered birth-death case. The similarity between~\eqref{ti1n} and~\eqref{feller2} allows to state that $u_{1, i + 1} = \delta T^i_{1, 1}$, which is reasonable as the former is the probability of going extinct on step $i + 1$ and the latter is the probability of being at $1$ on step $i$ multiplied by the death probability $\delta$. We can then use~\eqref{feller2}, the previous argument and  Corollary~\ref{zetacoro} to find alternate expressions for the survival probability, where $X \sim \text{Binom}(i, \beta)$ as usual:
    \begin{equation*}
        \P{N_i > 0 \, \middle| \, N_0 = 1} = 1 - \sum_{\substack{n = 1 \\ n \text{ odd}}}^i u_{1, n}(\beta, \delta; N) = 1 - \delta\sum_{\substack{m = 0 \\ m \text{ even}}}^{i - 1}T^m_{1, 1} = 1 - \delta\sum_{\ell = 0}^{\lfloor(i - 1) / 2\rfloor}\frac{\P{X = \ell}}{\ell + 1}.
    \end{equation*}
\end{rem}

\subsection*{Proofs of recurrence relations results}

Our approximate expected distributions depend on Taylor expansions of expected values of functions of random variables, along with the expected population size low-death approximation~\eqref{linearapprox}.
\begin{lem}\label{taylorexp}
    Consider a function $f \in C^3$ of two random variables $A$ and $B$ with finite expectations; the second-order Taylor expansion of $\E{f(A, B)}$ about $(\E{A}, \E{B})$ is
    \begin{equation}\label{secondordertaylorexp}
        \E{f(A, B)} = f(\E{A}, \E{B}) + \frac{1}{2}f_{AA}\vari(A) + f_{AB}\cov(A, B) + \frac{1}{2}f_{BB}\vari(B) + O^3,
    \end{equation}
    where the second derivatives $f_{ij} = \partial_i\partial_jf$ are evaluated at $(\E{A}, \E{B})$ and $O^3$ is big-O notation for any cubic terms.
\end{lem}
\begin{rem}
    The first-order expansion of $\E{f(A, B)}$ about $(\E{A}, \E{B})$ is $f(\E{A}, \E{B})$.
\end{rem}
\begin{coro}\label{secondordercoro}
    For $\E{B} \neq 0$ and $f(A, B) = A / B$, to second order we have
    \begin{equation*}\label{secondorderexp}
        \E{\frac{A}{B}} \simeq \frac{\E{A}}{\E{B}} - \frac{\cov(A, B)}{\E{B}^2} + \frac{\E{A}\vari(B)}{\E{B}^3} = \frac{\E{A}}{\E{B}} - \frac{\E{AB}}{\E{B}^2} + \frac{\E{A}\E{B^2}}{\E{B}^3}.
    \end{equation*} 
\end{coro}
We can now derive the birth-death site frequency spectrum, conditioned on survival.
\begin{prop}\label{sfsprop}
    Starting from one mutation-free progenitor cell, the birth-death process with new mutations arising with expectation $\mu$ has approximate expected site frequency spectrum
    \begin{equation}\label{sfs}
        \E{S_{j, i}\,\middle|\,N_0 = 1, N_i > 0} \simeq \sum_{j' = 0}^\infty \frac{2\mu(\delta / \beta)^{j'}}{(j + j')(j + j' + 1)}\E{N_i\,\middle|\,N_0 = 1, N_i > 0}.
    \end{equation}
\end{prop}
\begin{proof}
    As in the main text, we use the law of total expectation to write a recurrence relation for the SFS at step $i + 1$ conditional on the SFS at step $i$:
    \begin{align*}
        \E{S_{j, i + 1}} &= \E{\E{S_{j, i + 1}\,\middle|\,\{S_{j', i}\}_{j' \leq i}}}\\
        &= \E{S_{j, i} + \beta\left(-\frac{jS_{j, i}}{N_i} + \frac{(j - 1)S_{j - 1, i}}{N_i} + 2\mu\delta_{1,j}\right) + \delta\left(-\frac{jS_{j, i}}{N_i} + \frac{(j + 1)S_{j + 1, i}}{N_i}\right)},
    \end{align*}
    for $\delta_{\cdot, \cdot}$ the Kronecker delta symbol, and where we have omitted the conditioning on survival and the initial condition $N_0 = 1$ for brevity. The first term in each set of parentheses in the second line arises because any $j$-abundant mutations in the dividing (respectively dying) cell, the number of which is $jS_{j, i} / N_i$, will become $(j \pm 1)$-abundant and thus no longer contribute to the $j$-site. The second terms depict the corresponding contribution from any $(j \mp 1)$-abundant mutations in the dividing (respectively dying) cell, which will become $j$-abundant. The Kronecker delta source term represents the $2\mu$ new (thus $1$-abundant) mutations arising in any division.
    \begin{rem}
        For $N_0 = 1$, by observation, the expected population size and SFS at step $i = 1$ are $\E{N_1} = 2\beta$ and $\E{S_{1, 1}} = 2\mu\beta$, with the $\beta$ factors dropping if we condition on survival.
    \end{rem}
    In the pure-birth case, the missing telescoping calculation (for $Q_j = \E{S_{j, i}} / (i + 1)$) in the main text is
    \begin{equation*}
        Q_j = \frac{j - 1}{j + 1}Q_{j - 1} = \frac{j - 1}{j + 1}\frac{j - 2}{j}Q_{j - 2} = \cdots = \frac{2}{j(j + 1)}Q_1 = \frac{2\mu}{j(j + 1)},
    \end{equation*}
    which is indeed~\eqref{sfs} when $\delta = 0$ (as only the $j = 0$ term in the infinite sum remains). A proof similar to this was presented in~\cite{moellerMeasuresGeneticDiversification2022}.
    
    Now, for the birth-death case: if next we consider only the first-order approximation from~\eqref{secondordertaylorexp}
    \begin{equation*}
        \E{\frac{S_{j, i}}{N_i}\,\middle|\,N_0 = 1, N_i > 0} \simeq \frac{\E{S_{j, i}\,\middle|\,N_0 = 1, N_i > 0}}{\E{N_i\,\middle|\,N_0 = 1, N_i > 0}} = X_j,
    \end{equation*}
    where the final equality is the same ansatz as before (reasonable given our knowledge from~\cite{gunnarssonExactSiteFrequency2021a} that $\E{S_{j, i}\,\middle|\,N_0 = 1, N_i > 0} \propto \E{N_i\,\middle|\,N_0 = 1, N_i > 0}$). Again, we have a second-order non-homogeneous linear recurrence relation for $X_j$:
    \begin{equation*}
        X_j \E{N_{i + 1}\,\middle|\,N_0 = 1, N_i > 0} - X_j \E{N_i\,\middle|\,N_0 = 1, N_i > 0} = -j X_j + \beta(j - 1)X_{j - 1} + \delta(j + 1)X_{j + 1} + 2\mu\beta\delta_{1,j}.
    \end{equation*}
    For low death, where our first order approximation holds, from Figure~\ref{n-fig} and~\eqref{linearapprox} we note that the expected gain in population size in one time step is $\beta - \delta$. Thus the previous expression becomes
    \begin{equation}\label{xrecursion}
        \delta(j + 2)X_{j + 2} - (\beta - \delta + j + 1)X_{j + 1} + \beta jX_j = 0 \quad \text{ for } \quad j \geq 1,
    \end{equation}
    where we have absorbed the source term into the boundary conditions $X_1 = C$ and $X_2 = \beta(C - \mu) / \delta$.
    \begin{rem}
        One way to proceed is to use a shift operator $S$ (such that $SX_j = X_{j + 1}$) and factorise~\eqref{xrecursion}, solving first for one root of the ensuing quadratic equation in $S$ and then the other.
    \end{rem}
    Instead, we use Laplace's method~\cite{deijfenGrowingNetworksPreferential2009a} for solving a second-order recurrence relation of the form
    \begin{equation}\label{precursion}
        \left(\alpha_2(j + 2) + \beta_2\right)X_{j + 2} + \left(\alpha_1(j + 1) + \beta_1\right)X_{j + 1} + \left(\alpha_0 j + \beta_0\right)X_j = 0,
    \end{equation}
    which, in our case~\eqref{xrecursion}, we have $\{\alpha_i\}_{i = 0}^2 = \{\delta, -1, \beta\}$ and $\{\beta_i\}_{i = 0}^2 = \{0, \delta - \beta, 0\}$. To solve~\eqref{precursion} with appropriate boundary conditions we must find an interval $[a, b]$ and a function $h(t)$ such that
    \begin{equation*}
        \left[t^jh(t)\left(\alpha_2t^2 + \alpha_1t + \alpha_0\right)\right]_a^b = 0 \quad \text{ and } \quad \frac{h'(t)}{h(t)} = \frac{\beta_2t^2 + \beta_1t + \beta_0}{t(\alpha_2t^2 + \alpha_1t + \alpha_0)},
    \end{equation*}
    where $h(t)$ is defined up to a multiplicative constant determined by the boundary conditions. The solution to the recurrence relation~\eqref{precursion} is then
    \begin{equation}\label{pkexp}
        X_j \propto \int_a^b t^{j - 1}h(t) \dd t.
    \end{equation}
    Following~\cite{deijfenGrowingNetworksPreferential2009a}, we find $[a, b] = [0, 1]$ and
    \begin{equation*}\label{hexp}
        h(t) = \frac{1 - t}{\beta / \delta - t},
    \end{equation*}
    which is valid for $\delta > 0$. From a table of integrals~\cite{gradshteyn2007table} we have
    \begin{equation}\label{uvw}
        \int_0^1 x^{\lambda - 1}(1 - x)^{\mu - 1}(1 - \xi x)^{-\nu}\dd x = B(\lambda, \mu)\,{}_2 F_1(\nu, \lambda; \lambda + \mu; \xi)
    \end{equation}
    for $\Re(\lambda) > 0$, $\Re(\mu) > 0$ and $|\xi| < 1$, where we have defined the beta function $B(w, y)$ as
    \begin{equation*}
        B(w, y) = \int_0^1 s^{w - 1}(1 - s)^{y - 1}\dd s,
    \end{equation*}
    and the hypergeometric function ${}_2 F_1(k, l; m; z)$ as
    \begin{equation*}
        {}_2 F_1(k, l; m; z) = \sum_{n = 0}^\infty \frac{(k)_n (l)_n}{(m)_n}\frac{z^n}{n!},
    \end{equation*}
    for the rising Pochhammer symbol $(q)_n$
    \begin{equation*}
        (q)_n = \begin{cases}
        1 & \text{ if } n = 0,\\
        q(q + 1)\cdots(q + n - 1) &\text{ if } n > 0.
        \end{cases}
    \end{equation*}
    Substituting $\{\lambda, \mu, \nu, \xi\} \to \{j, 2, 1, \delta / \beta\}$ into~\eqref{uvw} and momentarily ignoring a constant factor of $\delta / \beta$ since $h(t)$ defined up to a constant,~\eqref{pkexp} becomes
    \begin{align*}
        X_j &\propto \int_0^1 t^{j - 1}\frac{1 - t}{1 - (\delta / \beta)t}\dd t\\
        &= B(j, 2)\,{}_2 F_1(1, j; j + 2; \delta / \beta)\\
        &= \left(\int_0^1 t^{j - 1}(1 - t)\dd t\right)\sum_{n = 0}^\infty \frac{(1)_n (j)_n}{(j + 2)_n}\frac{(\delta / \beta)^n}{n!}\\
        &= \left(\frac{1}{j} - \frac{1}{j + 1}\right)\sum_{n = 0}^\infty \frac{j(j + 1)}{(j + n)(j + n + 1)}\left(\frac{\delta}{\beta}\right)^n,
    \end{align*}
    which is~\eqref{sfs} after cancelling the factor $j(j + 1)$ and multiplying by $2\mu$ (so that $h(t)$ matches the boundary condition $X_1 = \E{S_{1, 1}} / \E{N_1} = \mu$).
\end{proof}
\begin{rem}
    The integral~\eqref{pkexp} is found via different methods in~\cite{gunnarssonExactSiteFrequency2021a} for the continuous time (stochastic population) case, though with bounds $[0, 1 - 1 / N]$. Thus in the large population limit these coincide, which aligns with our intuition that in the large population limit the continuous and discrete time cases should converge, as discussed in~\cite{gunnarssonExactSiteFrequency2021a}.
\end{rem}
We now derive the division distribution in the pure-birth and birth-death cases.
\begin{rem}
    For the expected pure-birth DD, the main text is simply missing an observation that the solution~\eqref{purebirthdd} solves the recurrence relation~\eqref{purebirthddrr}:
    \begin{equation*}
        \stir{i + 1}{\ell}\frac{2^\ell}{(i + 1)!} = \left(i\stir{i}{\ell} + \stir{i}{\ell - 1}\right)\frac{2^\ell}{(i + 1)!} = \frac{i}{i + 1}\stir{i}{\ell}\frac{2^\ell}{i!} + \frac{2}{i + 1}\stir{i}{\ell - 1}\frac{2^{\ell - 1}}{i!}.
    \end{equation*}
\end{rem}
\begin{prop}\label{ddfirstorderprop}
    For the birth-death process described in Figure~\ref{chain-tree-fig}a, the expected division distribution is approximated to first order by
    \begin{equation}\label{ddfirstorder}
        \E{D_{\ell, i}\,\middle|\,N_0 = 1, N_i > 0} \simeq \frac{\stir{i}{\ell}2^\ell\left(1 - \delta / \beta\right)^{-\ell}}{\sum_{\ell' = 1}^i\stir{i}{\ell'}2^{\ell'}\left(1 - \delta / \beta\right)^{-\ell'}}\E{N_i\,\middle|\,N_0 = 1, N_i > 0}.
    \end{equation}
\end{prop}
\begin{proof}
    We employ our same approach as for the pure-birth case:
    \begin{align*}
        \E{D_{\ell, i + 1}} &= \E{\E{D_{\ell, i + 1}\,\middle|\,\{D_{\ell', i}\}_{\ell' \leq i}, N_{i + 1} > 0}}\\
        &= \E{D_{\ell, i} + \beta\left(-\frac{1}{N_i}D_{\ell, i} + \frac{2}{N_i}D_{\ell - 1, i}\right) + \delta\left(-\frac{1}{N_i}D_{\ell, i}\right)}\\
        &\simeq \left(1 - \frac{1}{\E{N_i}}\right)\E{D_{\ell, i}} + \frac{2\beta}{\E{N_i}}\E{D_{\ell - 1, i}},
    \end{align*}
    up to first order, where in the second line the impact of a cell dying is simply to reduce its contribution to the division distribution by its expected amount. As before, we have omitted the conditioning on survival and the initial condition $N_0 = 1$ for brevity.

    Next, consider the ansatz
    \begin{equation}\label{dijans}
        \E{D_{\ell, i}}_\text{az} = \frac{\stir{i}{\ell}(2\beta)^\ell(\beta - \delta)^{i - \ell}}{\prod_{i' = 1}^{i - 1}\E{N_{i'}}},
    \end{equation}
    which satisfies the previous recurrence relation by noting that, as argued in the remark following Proposition~\ref{expectedniprop} (as well as by~\eqref{linearapprox} and Figure~\ref{n-fig}), $\E{N_i\,\middle|\,N_0 = 1, N_i > 0} - 1 \simeq (\beta - \delta)i$ in the low-death limit:
    \begin{align*}
        \E{D_{\ell, i + 1}}_\text{az} &= \frac{\stir{i + 1}{\ell}(2\beta)^\ell(\beta - \delta)^{i + 1 - \ell}}{\prod_{i' = 1}^i\E{N_{i'}}}\\
        &= \left(i\stir{i}{\ell} + \stir{i}{\ell - 1}\right)\frac{(2\beta)^\ell(\beta - \delta)^{i + 1 - \ell}}{\E{N_i}\prod_{i' = 1}^{i - 1}\E{N_{i'}}}\\
        &= \frac{\E{N_i} - 1}{\E{N_i}}\stir{i}{\ell}\frac{(2\beta)^\ell(\beta - \delta)^{i - \ell}}{\prod_{i' = 1}^{i - 1}\E{N_{i'}}} + \frac{2\beta}{\E{N_i}}\stir{i}{\ell - 1}\frac{(2\beta)^{\ell - 1}(\beta - \delta)^{i - (\ell - 1)}}{\prod_{i' = 1}^{i - 1}\E{N_{i'}}}.
    \end{align*}
    It remains to show that the ansatz~\eqref{dijans} can be rewritten as~\eqref{ddfirstorder}. In this pursuit, we employ the following lemma.
    \begin{lem}\label{stirlinglem}
    For $x \neq 0$, the unsigned Stirling numbers defined by~\eqref{stirlingdef} satisfy
    \begin{equation}\label{stirlinglem2}
        \prod_{k = 1}^i\left(\frac{k}{x} + 1\right) = x^{-i}\sum_{k = 1}^i \stir{i}{k}(1 + x)^k.
    \end{equation}
    \end{lem}
    \begin{proof}[Proof of Lemma~\ref{stirlinglem}]
        Note that both expressions coincide to $1 + 1 / x$ when $i = 1$; suppose they do for $1 \leq i \leq I$. For $i = I + 1$, the right-hand side of~\eqref{stirlinglem2} is
        \begin{align*}
            x^{-(I + 1)}\sum_{k = 1}^{I + 1} \stir{I + 1}{k}(1 + x)^k &= x^{-(I + 1)}\sum_{k = 1}^{I + 1} \left(I\stir{I}{k} + \stir{I}{k - 1}\right)(1 + x)^k\\
            &= x^{-(I + 1)}\left(\sum_{k = 1}^I I\stir{I}{k}(1 + x)^k + \sum_{k = 1}^I\stir{I}{k}(1 + x)^{k + 1}\right)\\
            &= \frac{1}{x}\left(I + 1 + x\right)x^{-I}\sum_{k = 1}^I \stir{I}{k}(1 + x)^k\\
            &= \left(\frac{I + 1}{x} + 1\right)\prod_{k = 1}^I\left(\frac{k}{x} + 1\right),
        \end{align*}
        where we have used the definition~\eqref{stirlingdef} and the boundary conditions of the unsigned Stirling numbers of the first kind, resulting the left-hand expression of~\eqref{stirlinglem2}, concluding the proof by induction.
    \end{proof}
    Applying Lemma~\ref{stirlinglem} with $x = 1 / (\beta - \delta)$,~\eqref{dijans} transforms into
    \begin{equation*}
        \frac{\stir{i}{\ell}(2\beta)^\ell(\beta - \delta)^{i - \ell}}{\prod_{i' = 1}^{i - 1}\E{N_{i'}}} = \frac{\stir{i}{\ell}(2\beta)^\ell(\beta - \delta)^{i - \ell}}{(\beta - \delta)^i\sum_{\ell' = 1}^i \stir{i}{\ell'}(1 + 1 / (\beta - \delta))^{\ell'}}\E{N_i},
    \end{equation*}
    which is the desired expression~\eqref{ddfirstorder} after the identifications
    \begin{equation*}
        \frac{2\beta}{\beta - \delta} = 1 + \frac{1}{\beta - \delta} = 2\left(1 - \frac{\delta}{\beta}\right)^{-1}.\qedhere
    \end{equation*}
\end{proof}
\begin{rem}
    When the distributions of $D_{\ell, i}$ and $N_i$ are not sufficiently clustered around their means (that is, when $\delta$ isn't small), however, the first-order approximation for the division distribution is not sound. As both diverge as $i \to \infty$, we require their variances to grow more slowly than their means do.
\end{rem}
\begin{prop}\label{varianceprop}
    The variance of the population size $N_i$ for the birth-death process with $N_0 = 1$, using the low-death approximation $\E{N_i\,\middle|\,N_0 = 1, N_i > 0} \simeq (\beta - \delta)i + 1$, satisfies
    \begin{equation*}
        \mathrm{var}\left(N_i\right) \simeq \left(1 - (\beta - \delta)^2\right)i,
    \end{equation*}
    and thus tends to $0$ as $\delta \to 0$ (since then $\beta \to 1$), as expected.
\end{prop}
\begin{proof}
    We omit the conditioning on survival and the initial condition $N_0 = 1$ for brevity. Using our usual approach, we have
    \begin{equation*}
        \E{N_{i + 1}^2} = \E{\E{N_{i + 1}^2\,\middle|\,N_i}} = \E{N_i^2} + 2(\beta - \delta)\E{N_i} + 1 \simeq \E{N_i^2} + 2(\beta - \delta)\left((\beta - \delta)i + 1\right) + 1.
    \end{equation*}
    Solving the ensuing recurrence, we find
    \begin{equation}\label{nsquaredexp}
        \E{N_i^2} \simeq (\beta - \delta)^2i^2 + \left(-(\beta - \delta)^2 + 2(\beta - \delta) + 1\right)i + 1,
    \end{equation}
    from which we subtract $\E{N_i}^2 \simeq \left((\beta - \delta)i + 1\right)^2$:
    \begin{equation*}
        \mathrm{var}\left(N_i\right) \simeq (\beta - \delta)^2i^2 + \left(-(\beta - \delta)^2 + 2(\beta - \delta) + 1\right)i + 1 - \left((\beta - \delta)i + 1\right)^2 = \left(1 - (\beta - \delta)^2\right)i.\qedhere
    \end{equation*}
\end{proof}
\begin{rem}
    There is an equivalent recurrence relation for the division distribution when expanded to second order using Corollary~\ref{secondordercoro}, though it cannot be solved by elementary methods:
    \begin{align*}
        \E{D_{\ell, i + 1}} &= \E{D_{\ell, i}} - \E{\frac{D_{\ell, i}}{N_i}} + 2\beta\E{\frac{D_{\ell - 1, i}}{N_i}}\\
        &\simeq \E{D_{\ell, i}} - \left(\frac{\E{D_{\ell, i}}}{\E{N_i}} - \frac{\E{D_{\ell, i}N_i}}{\E{N_i}^2} + \frac{\E{D_{\ell, i}}\E{N_i^2}}{\E{N_i}^3}\right)\\
        &\quad + 2\beta\left(\frac{\E{D_{\ell - 1, i}}}{\E{N_i}} - \frac{\E{D_{\ell - 1, i}N_i}}{\E{N_i}^2} + \frac{\E{D_{\ell - 1, i}}\E{N_i^2}}{\E{N_i}^3}\right)\\
        &= \E{D_{\ell, i}} - \frac{\E{D_{\ell, i}}}{\E{N_i}} + 2\beta\frac{\E{D_{\ell - 1, i}}}{\E{N_i}} + \frac{\E{N_i^2}}{\E{N_i}^3}\left(2\beta\E{D_{\ell - 1, i}} - \E{D_{\ell, i}}\right)\\
        &\quad + \frac{1}{\E{N_i}^2}\sum_{\ell' = 1}^i\left(\E{D_{\ell, i}D_{\ell', i}} - 2\beta\E{D_{\ell - 1, i}D_{\ell', i}}\right),
    \end{align*}
    where we used the linearity of expectation to simplify and have omitted all conditioning on survival and the initial condition $N_0 = 1$ for brevity, as usual. The first three terms in the final expression are those from the first-order expansion of Proposition~\ref{ddfirstorderprop}, and the next term can be simplified using~\eqref{nsquaredexp}. Despite this, the equation above for $\E{D_{\ell, i + 1}}$ not only remains complicated but includes terms of the form $\E{D_{\ell, i}D_{\ell', i}}$, which prevent it from being solved using our usual recursive methods.
\end{rem}

\subsection*{Mutational occurrences}

Recall the definition of the number of mutational occurrences from~\eqref{sfsmbdmirror}; we will write $C_i$ for this quantity. Note that this is also equal to (up to renormalisation by the number of unique mutations $M_i$ or the population $N_i$) the mean of the site frequency spectrum $\{S_{j, i}\}_j$ or the single-cell mutational burden distribution $\{B_{k, i}\}_k$, respectively. In the pure-birth case, we have the following result.
\begin{prop}\label{umbeta1prop}
    In a pure-birth process with mutation rate $\mu$ starting from a single mutation-free progenitor cell, the expected number of unique mutations and mutational occurrences at step $i$ are given by
    \begin{equation*}
        \E{M_i} = 2\mu i \quad \text{ and } \quad \E{C_i} = 2\mu(i + 1)\left(H_{i + 1} - 1\right),
    \end{equation*}
    for $H_n = \sum_{k = 1}^n k^{-1}$ the $n$th harmonic number.
\end{prop}
\begin{proof}
With each birth, an expected $2\mu$ new mutations are added to the population, so the result for $\E{M_i}$ is clear by observation. Taking the sum over $j$ of the expected pure-birth SFS~\eqref{purebirthsfs} as in the definition $M_i = \sum_j S_{j, i}$, and using the linearity of expectation and $j^{-1}(j + 1)^{-1} = j^{-1} - (j + 1)^{-1}$ to telescope terms, we also obtain the desired expression $\E{M_i} = 2\mu i$.

Using the usual law of total expectation approach, the recurrence relation for $C_i$ is
\begin{equation}\label{urecursionbeta1}
    \E{C_{i + 1}} = \E{\left(1 + \frac{1}{N_i}\right)C_i + 2\mu},
\end{equation}
since not only are $2\mu$ new mutations arising with each division event, but all of the mutations in the dividing cell are duplicated. The proportion of cells with $k$ mutations is $B_{k, i} / N_i$, so the expected number of mutations in the dividing cell is found by multiplying this proportion by $k$ and summing over $1 \leq k \leq M_i$. This sum is exactly $C_i / N_i$ by~\eqref{sfsmbdmirror}, giving rise to the parenthetical term in~\eqref{urecursionbeta1}.

Recalling that in the pure-birth case we have $N_i = i + 1$,~\eqref{urecursionbeta1} can be solved recursively in a similar manner to the pure-birth case of Proposition~\ref{sfsprop}: we divide both sides of~\eqref{urecursionbeta1} by $i + 2$, make a change of variables to $G_i = C_i / (i + 1)$ such that $G_1 = C_1 / 2 = \mu$, telescope the $G_i$ terms to find
\begin{equation*}
     G_{i + 1} = G_1 + 2\mu\sum_{k = 1}^i\frac{1}{k + 2} = \mu + 2\mu\left(H_{i + 2} - 1 - \frac{1}{2}\right) = 2\mu\left(H_{i + 2} - 1\right),
\end{equation*}
which, after multiplying by $i + 2$ and relabelling indices, is the desired expression.
\end{proof}
\begin{prop}\label{umprop}
For a birth-death process with mutation rate $\mu$, the expected number of unique mutations $\E{M_i}$ to first order has the following form:
\begin{equation}\label{msol}
    \E{M_i} \simeq 2\mu\left(\beta - \frac{\beta(\beta - \delta)}{\delta}\log\left(1 - \frac{\delta}{\beta}\right)(i - 1)\right)
\end{equation}
and the expected total mutational occurrences $\E{C_i}$ to first order obey
\begin{equation}\label{usol}
    \E{C_i} \simeq 2\mu\beta \E{N_i}\sum_{i' = 1}^i\frac{1}{\E{N_{i'}}},
\end{equation}
where all expectations are conditioned on survival and the initial condition $N_0 = 1$.
\end{prop}
\begin{proof}
First we will find the expression~\eqref{msol} for $\E{M_i}$. Note that the previous contributions during birth events of the pure-birth case are simply weighted by $\beta$; the death events mean that $1$-abundant mutations found in the dying cell (whose expected number is $\E{S_{1, i} / N_i}$, which we will expand to first order as before) vanish from the count at step $i + 1$. Thus, the recurrence relation for $\E{M_i}$ is
\begin{equation*}
    \E{M_{i + 1}} = \E{M_i + 2\mu\beta - \delta\frac{S_{1, i}}{N_i}},
\end{equation*}
which is solved by summing over $i$ and telescoping to find
\begin{equation}\label{mintermediary}
    \E{M_i} \simeq 2\mu\beta i - \delta\sum_{i' = 1}^{i - 1}\frac{\E{S_{1, i'}}}{\E{N_{i'}}}= 2\mu\beta i - \delta(i - 1)\sum_{n = 0}^\infty \frac{2\mu(\delta / \beta)^n}{(n + 1)(n + 2)},
\end{equation}
where we have taken advantage of the fact that (under our approximations) $\E{S_{1, i}} / \E{N_i}$ is independent of $i$ and substituted~\eqref{sfs}. A crude approximation would be this expression for $i' = 1$: $\E{S_{1, 1}} / \E{N_1} = \mu$. However, our approximations hold better for large $i$, so instead we compute
\begin{equation*}
    \sum_{n = 0}^\infty \frac{w^n}{(n + 1)(n + 2)} = 1 - \left(1 - w\right)\sum_{m = 0}^\infty\frac{w^m}{m + 2} = 1 + \frac{1 - w}{w^2} - \frac{1 - w}{w^2}\sum_{m = 1}^\infty\frac{w^m}{m} = \frac{1}{w} + \frac{1 - w}{w^2}\log(1 - w),
\end{equation*}
which we substitute into~\eqref{mintermediary} with $w = \delta / \beta$ to obtain
\begin{equation*}
    \E{M_i} \simeq 2\mu\beta i - 2\mu\delta(i - 1)\left(\frac{\beta}{\delta} + \left(1 - \frac{\delta}{\beta}\right)\frac{\beta^2}{\delta^2}\log\left(1 - \frac{\delta}{\beta}\right)\right),
\end{equation*}
which can be rearranged to find~\eqref{msol}.

To find the expression~\eqref{usol} for the total mutational occurrences, note that as before the birth contributions are weighted by $\beta$; the death contributions to $C_{i + 1}$ include the previous quantity $C_i$, minus the number of mutations found in the dying cell, whose expected number is $\E{C_i / N_i}$, as argued in the proof of Proposition~\ref{umbeta1prop}. Thus, the recurrence relation satisfied by $\E{C_i}$ is
\begin{equation*}
    \E{C_{i + 1}} = \E{\left(1 + \frac{\beta - \delta}{N_i}\right)C_i + 2\mu\beta}.
\end{equation*}
We notice that by~\eqref{linearapprox} the parenthetical term expands to first order to $\E{N_{i + 1}} / \E{N_i}$, so by dividing by $\E{N_{i + 1}}$ and solving for $\E{C_i} / \E{N_i}$ (by the same telescoping trick as in the proof of Proposition~\ref{umbeta1prop}), we obtain
\begin{equation*}
    \frac{\E{C_{i + 1}}}{\E{N_{i + 1}}} \simeq \frac{\E{C_1}}{\E{N_1}} + 2\mu\beta\sum_{i' = 1}^i\frac{1}{\E{N_{i' + 1}}} = \mu\left(1 + 2\beta\left(\sum_{i' = 1}^{i + 1}\frac{1}{\E{N_{i'}}} - \frac{1}{2\beta}\right)\right) = 2\mu\beta\sum_{i' = 1}^{i + 1}\frac{1}{\E{N_{i'}}},
\end{equation*}
which gives the desired expression after multiplying both sides by $\E{N_{i + 1}}$ and relabelling indices.
\end{proof}

\subsection*{Other mutational distributions}

As noted in the main text, when other mutational distributions from the Poisson distribution are used in simulations, the conversion from the division distribution (DD) to the single-cell mutational burden distribution (MBD) described by Figure~\ref{mbd-fig} still holds under some conditions. In Figure~\ref{unif-geom-fig}, the mutational distributions used are $\text{Unif}\{1, \dots, 2\mu - 1\}$ and $\text{Geom}(1 / \mu)$, both of which have mean $\mu$ and are recaptured by the predicted distribution.

\begin{figure}[ht]
    \centering
    \includegraphics[width = \textwidth]{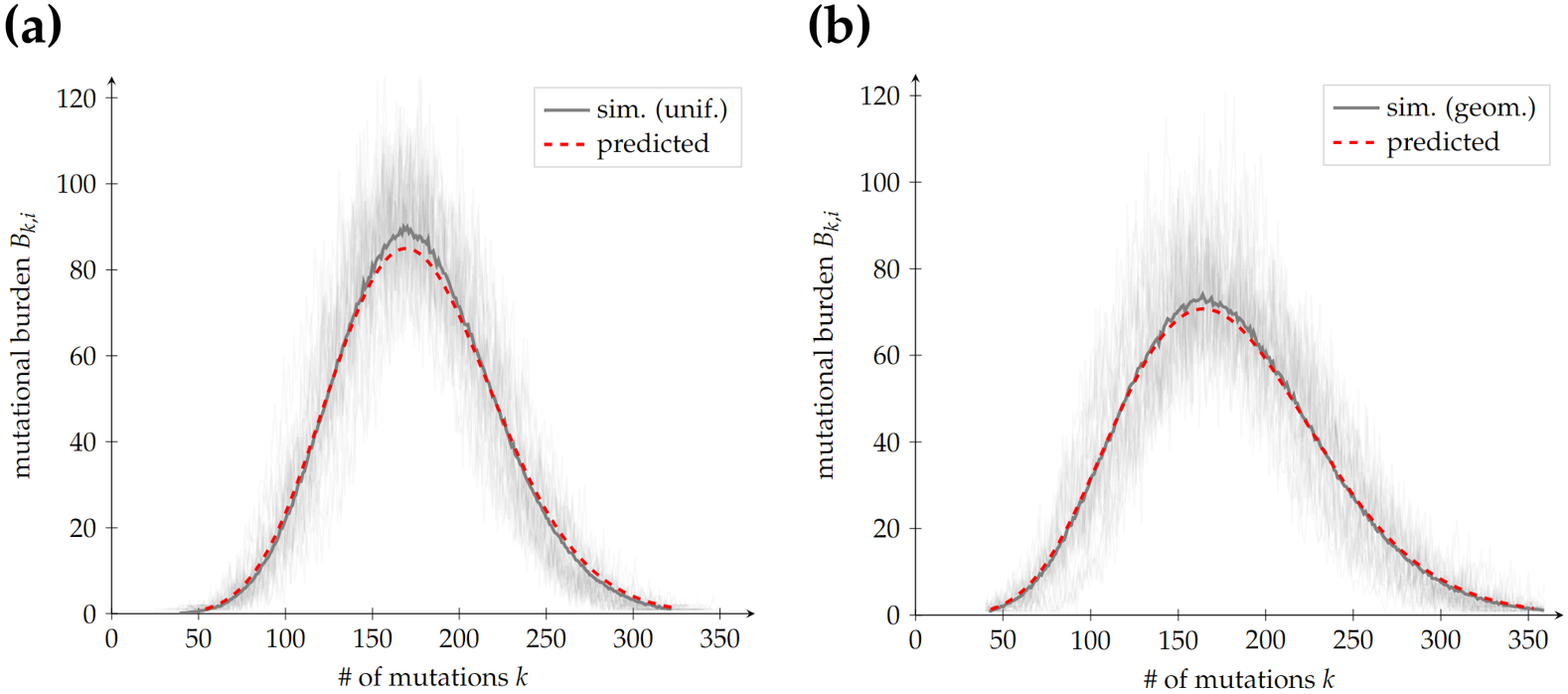}
    \caption{Conversions from the division distribution to the single-cell mutational burden distribution as in Figure~\ref{mbd-fig} of the main text, using other (than Poisson) mutational distributions, with means $\mu$: \textbf{(a)}~uniform distribution over the set $\{1, \dots, 2\mu - 1\}$ and \textbf{(b)}~geometric distribution with mean $\mu$. Average (solid dark grey line) of $200$ simulation realisations (representatives in solid pale grey lines) and the predicted MBD distribution converted from the DD (dashed red line).}
    \label{unif-geom-fig}
\end{figure}

If, however, the number of new mutations $U_1$ and $U_2$ (where the indices refer to daughter cells) are drawn from a distribution with small support, such as the uniform distribution $\text{Unif}\{\mu - 1, \mu + 1\}$, then it is unsurprising that the smoothing procedure from a DD to a MBD of Figure~\ref{mbd-fig} doesn't hold, since fewer possible mutational burdens are obtainable for cells. In the extreme case of the delta distribution $\text{Delta}(\mu)$, we simply obtain a scaling of the DD, where all cells have mutational burden equal to the product of the mutational mean and their division burden.

On the other hand, given an expected MBD, it is possible to recover the expected DD via a binning procedure. By histogramming the MBD with bins of width $\mu$ centred at integer multiples of $\mu$ (which results in the distribution supposing that exactly $U_1 = U_2 = \mu$ mutations were acquired with each division), we can then rescale by $\mu$ to obtain the DD. Symbolically,
\begin{equation}\label{histogrammingexpectations}
    \E{D_{\ell, i}} = \sum_{k \in \left[\left(\ell - \frac{1}{2}\right)\mu, \left(\ell + \frac{1}{2}\right)\mu\right)} \E{B_{k, i}}.
\end{equation}
In Figure~\ref{mbd-fig}a, this involves summing $\mu$ adjacent bins of the MBD, and then rescaling the $x$-axis by $\mu$ to obtain the DD. This is exactly the aforementioned distribution, obtained by using $\text{Delta}(\mu)$ as the mutational distribution.
\end{document}